\documentclass[conference]{IEEEtran}
\IEEEoverridecommandlockouts
\usepackage{cite}
\usepackage{amsmath,amssymb,amsfonts}
\usepackage{algorithmic}
\usepackage{graphicx}
\usepackage{textcomp}
\usepackage{xcolor}
\usepackage[ruled,vlined]{algorithm2e}
\usepackage{amsthm}
\usepackage{multirow}
\usepackage{subcaption}
\usepackage{etoolbox}
\usepackage{enumitem}
\newcommand{\minitab}[2][l]{\begin{tabular}{#1}#2\end{tabular}}
\makeatletter
\newcommand\niton{\mathrel{\m@th\mathpalette\canc@l\owns}}
\newcommand\canc@l[2]{{\ooalign{$\hfil#1/\mkern1mu\hfil$\crcr$#1#2$}}}
\makeatother
\makeatletter
\patchcmd{\@algocf@start}
  {-1.5em}
  {-0.3em}
  {}{}
\makeatother

\newcommand{\textoverline}[1]{$\overline{\mbox{#1}}$}

\newtheorem{definition}{Definition}
\newtheorem{theorem}{Theorem}
\newtheorem{example}{Example}
\AtEndEnvironment{example}{\null\hfill\qedsymbol}
\renewcommand{\arraystretch}{1.25}
\let\oldexample\example
\renewcommand{\example}{\oldexample\normalfont}

\newtheoremstyle{cited}%
  {3pt}
  {3pt}
  {\itshape}
  {}
  {\bfseries}
  {.}
  {.5em}
  {\thmname{#1} \thmnumber{#2} \thmnote{\normalfont#3}}
\theoremstyle{cited}

\newtheorem{lemma}{Lemma}
\graphicspath{{./Figures/}}

\def\BibTeX{{\rm B\kern-.05em{\sc i\kern-.025em b}\kern-.08em
    T\kern-.1667em\lower.7ex\hbox{E}\kern-.125emX}}
    
\begin{document}
\title{Regular Path Query Evaluation Sharing a Reduced Transitive Closure Based on Graph Reduction}
\author{
\IEEEauthorblockN{Inju Na}
\IEEEauthorblockA{\textit{School of Computing} \\
\textit{KAIST}\\
Daejeon, Republic of Korea\\
ijna@kaist.ac.kr}
\and
\IEEEauthorblockN{Yang-Sae Moon\textsuperscript{*}}
\IEEEauthorblockA{\textit{Department of Computer Science} \\
\textit{Kangwon National University}\\
Chuncheon, Republic of Korea\\
ysmoon@kangwon.ac.kr}
\thanks{\textsuperscript{*}A corresponding author.}
\and
\IEEEauthorblockN{Ilyeop Yi, Kyu-Young Whang, Soon J. Hyun}
\IEEEauthorblockA{\textit{School of Computing} \\
\textit{KAIST}\\
Daejeon, Republic of Korea\\
\{iyyi, kywhang, sjhyun\}@kaist.ac.kr}
}

\maketitle
This work has been submitted to the IEEE for possible publication. Copyright may be transferred without notice, after which this version may no longer be accessible.\\\\\\
\begin{abstract}
Regular path queries (RPQs) find pairs of vertices of paths satisfying given regular expressions on an edge-labeled, directed multigraph. When evaluating an RPQ, the evaluation of a Kleene closure (i.e., Kleene plus or Kleene star) is very expensive. Furthermore, when multiple RPQs include a Kleene closure as a common sub-query, repeated evaluations of the common sub-query cause serious performance degradation. In this paper, we present a novel concept of \textit{RPQ-based graph reduction}, which significantly simplifies the original graph through edge-level and vertex-level reductions. Interestingly, RPQ-based graph reduction can replace the evaluation of the Kleene closure on the large original graph to that of the transitive closure to the small reduced graph. We then propose a reduced transitive closure (RTC) as a lightweight structure for efficiently sharing the result of a Kleene closure. We also present an RPQ evaluation algorithm, \textit{RTCSharing}, which treats each clause in the disjunctive normal form of the given RPQ as a batch unit. If the batch units include a Kleene closure as a common sub-query, we share the lightweight RTC instead of the heavyweight result of the Kleene closure. RPQ-based graph reduction further enables us to formally represent the result of an RPQ including a Kleene closure as a relational algebra expression including the RTC. Through the formal expression, we optimize the evaluation of the batch unit by eliminating \textit{useless} and \textit{redundant operations} of the previous method. Experiments show that RTCSharing improves the performance significantly by up to 73.86 times compared with existing methods in terms of query response time.
\end{abstract}

\begin{IEEEkeywords}
regular path queries (RPQ), query optimization, graph reduction, reduced transitive closure (RTC), RTCSharing
\end{IEEEkeywords}

\section{Introduction}
A regular path query (RPQ) is a regular expression that finds ordered pairs of vertices of paths satisfying the RPQ on an edge-labeled, directed multigraph~\cite{Men95}. Here, paths satisfy an RPQ when a sequence of labels matches the RPQ. The RPQ is attracting attention as an essential and important operation for graph data such as social networks and the Semantic Web and is supported by typical graph query languages such as SPARQL 1.1~\cite{Spa13} and Cypher~\cite{Cyp18}. RPQs can be used in a variety of applications such as signal path detection in protein networks, recommending friends in social networks, and extracting information from linked open data~\cite{Ngu17,Yak16}. 

To evaluate an RPQ, it is necessary to traverse the graph from each vertex and perform pattern matching for labels of edges accessed during the traversal, finding the paths that satisfy the query. As such, evaluating an RPQ is a complex operation that combines graph traversal and pattern matching~\cite{Gra03a}. In particular, when an RPQ includes a Kleene closure (i.e., Kleene plus or Kleene star), the evaluation is more expensive because there can be long paths of varying lengths satisfying the RPQ~\cite{Gra03b}. Therefore, optimization of evaluating RPQs is actively being studied as an important issue~\cite{Abu17,Fle16,Kos12,Ngu17,Yak15,Yak16}. Optimization of evaluating RPQs needs more attention when evaluating multiple RPQs. When there is a common sub-query among multiple RPQs, evaluating these RPQs individually leads to repeated evaluations of the common sub-query. This problem of repeated evaluations can be solved by evaluating the common sub-query once and sharing the result among the RPQs. However, when a common sub-query includes a Kleene closure, evaluating the common sub-query itself is expensive.

In this paper, we present a novel concept of \textit{RPQ-based graph reduction}, which converts an original labeled multigraph \textit{G} to an unlabeled simple graph \textit{\textoverline{G\textsubscript{R}}}. RPQ-based graph reduction consists of two levels: 1) edge-level reduction producing an unlabeled graph \textit{G\textsubscript{R}} from \textit{G} by mapping paths satisfying \textit{R} on \textit{G} to edges of \textit{G\textsubscript{R}} and 2) vertex-level reduction producing a much reduced graph \textit{\textoverline{G\textsubscript{R}}} from \textit{G\textsubscript{R}} by representing each strongly connected component (SCC) as one vertex, which we define in Section 3. Interestingly, RPQ-based graph reduction can replace the evaluation of the Kleene closure on the large original graph \textit{G} to that of the transitive closure to the small reduced graph \textit{\textoverline{G\textsubscript{R}}}. Even though the graph reduction incurs a little overhead, the performance gain of using the reduced graph is much larger than the overhead. Based on the RPQ-based graph reduction, we then propose a reduced transitive closure (RTC) as a lightweight structure for efficiently sharing the result of a Kleene closure. Let \textit{R\textsuperscript{+}\hspace{-0.12cm}\textsubscript{G}} (or \textit{R\textsuperscript{*}\hspace{-0.12cm}\textsubscript{G}}) be the evaluation result of a Kleene closure \textit{R\textsuperscript{+}} (or \textit{R\textsuperscript{*}}) for any regular expression \textit{R} on the graph \textit{G}. Since \textit{R\textsuperscript{*}\hspace{-0.12cm}\textsubscript{G}} can be easily derived from \textit{R\textsuperscript{+}\hspace{-0.12cm}\textsubscript{G}}, we hereafter use \textit{R\textsuperscript{+}} and \textit{R\textsuperscript{+}\hspace{-0.12cm}\textsubscript{G}} only unless confusion occurs. In the final algorithms (Algorithms 1 and 2), we also deal with \textit{R\textsuperscript{*}} by simply deriving it from \textit{R\textsuperscript{+}}. In Theorem 1, we formally show that \textit{R\textsuperscript{+}\hspace{-0.12cm}\textsubscript{G}}, the evaluation result of \textit{R\textsuperscript{+}} on the original graph \textit{G}, can be easily calculated from the transitive closure (i.e., RTC) of the reduced graph \textit{\textoverline{G\textsubscript{R}}}.

We also propose an RPQ evaluation algorithm, Reduced Transitive Closure Sharing (\textit{RTCSharing}), which converts the given RPQ to a logically equivalent disjunctive normal form (DNF) and evaluates it treating each clause as a batch unit. As in~\cite{Dav90}, we can convert all RPQs to a logically equivalent DNF treating each outermost Kleene closure as a literal. If multiple batch units include a Kleene closure as a common sub-query, we share the lightweight RTC instead of the heavyweight result of the Kleene closure. RPQ-based graph reduction further enables us to represent the result of an RPQ including a Kleene closure as a relational algebra expression in the form of a join sequence including the RTC. By representing the batch unit as a relational algebra expression, RTCSharing optimizes the evaluation of the batch unit by eliminating \textit{redundant} and \textit{useless operations} (see Section 4.2 for formal definitions) as follows: (1) it eliminates redundant operations (in effect, mostly duplicate check operations) that might occur in the next step removing redundant elements by unioning the intermediate results of each join step; and (2) it also eliminates useless operations caused by selection through \textit{Prefix} and by the property of the reduced graph \textit{\textoverline{G\textsubscript{R}}}. 

The contributions of this paper are summarized as follows:

\begin{itemize}
    \item We present a novel notion of \textit{RPQ-based graph reduction} that converts a labeled multigraph to an unlabeled simple graph for the efficient evaluation of a Kleene closure. 
    
    \item We propose an RTC as a lightweight structure to efficiently share the result of a common sub-query among RPQs whose common sub-query is a Kleene closure \textit{R\textsuperscript{+}} for any regular expression \textit{R}.
    
    \item We show that \textit{R\textsuperscript{+}} can be evaluated as the transitive closure of the edge-level reduced graph \textit{G\textsubscript{R}} (Lemma~\ref{lemma:RstarG equals RTC of GR}), and further, can be calculated from the transitive closure of the two-level reduced graph \textit{\textoverline{G\textsubscript{R}}} (i.e., the RTC) (Lemma~\ref{lemma:RTC of GR equals Cartesian product of RTC of GRbar} and Theorem~\ref{theorem:R+G equals Cartesian product of RTC of GRbar}). 
    
    \item We propose an RPQ evaluation algorithm, \textit{RTCSharing}, that uses each clause in the DNF of the given RPQ as a batch unit. RTCSharing shares the RTC instead of \textit{R\textsuperscript{+}\hspace{-0.12cm}\textsubscript{G}} among batch units as the result of the common sub-query \textit{R\textsuperscript{+}}.
    
    \item We show that the result of an RPQ including a Kleene closure can be represented as a relational algebra expression including the RTC and exploit it for optimizing the evaluation of the RPQ by eliminating \textit{useless} and \textit{redundant operations}.

    \item Through experiments with synthetic and real datasets, we show that \textit{RTCSharing} significantly outperforms existing algorithms: 1) the recent algorithm that shares \textit{R\textsuperscript{+}\hspace{-0.12cm}\textsubscript{G}} among RPQs\cite{Abu17} and 2) the naive algorithm that shares nothing among RPQs\cite{Yak16}.
\end{itemize}

The paper is organized as follows. Section~\ref{sec:Preliminaries} describes the preliminary background. Section~\ref{sec:Intermediate Representation} proposes RPQ-based graph reduction and the RTC. Section~\ref{sec:Multiple Regular Path Queries(MRPQs) Evaluation Method} describes the \textit{RTCSharing} algorithm with the optimization. Section~\ref{sec:Performance Evaluation} shows the results of performance evaluation for the proposed methods in comparison with existing RPQ evaluation algorithms. Section~\ref{sec:Related Work} summarizes the related work. Section~\ref{sec:Conclusion} concludes the paper.

\section{Preliminaries}
\label{sec:Preliminaries}
In this section, we introduce the background of regular path query research. Section~\ref{subsec:Data Model} describes the data model used for regular path queries. Section~\ref{subsec:Regular Path Query(RPQ)} describes the definition of regular path queries and methods of evaluating them. Section~\ref{subsec:Evaluation Method of Multiple RPQs} describes methods of evaluating multiple regular path queries. 

\subsection{Data Model}
\label{subsec:Data Model}
The data used for regular path queries is an edge-labeled, directed multigraph \textit{G}, which is defined as a 5-tuple (\textit{V}, \textit{E}, \textit{f}, $\mathit{\Sigma}$, \textit{l})~\cite{Men95}. \textit{V} is a set of vertices. \textit{E} is a set of directional edges. \textit{f}: \textit{E} $\rightarrow$ \textit{V}$\times$\textit{V} is a function that maps each edge to an ordered pair of two vertices connected by the edge. $\mathit{\Sigma}$ is a set of labels. \textit{l}: \textit{E} $\rightarrow$ $\mathit{\Sigma}$ is a function that maps each edge to its label. In \textit{G}, multiple edges between any two vertices are allowed, but the labels of these edges must be distinct. Fig.~\ref{fig:an example of graph} shows a graph, which will be used as an example graph throughout the paper.

\begin{figure}[b]
\centering
\includegraphics[scale= 0.88]{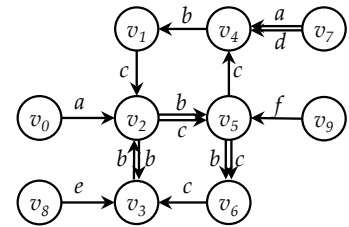}
\caption{An edge-labeled, directed multigraph.}
\label{fig:an example of graph}
\end{figure}

TABLE~\ref{tab:notations for G} summarizes the notation related to \textit{G} to be used in this paper. Each vertex has a unique ID (VID). A vertex with VID \textit{i} is denoted by \textit{v\textsubscript{i}}. A label has a unique ID (LID). A label with LID \textit{i} is denoted by \textit{l\textsubscript{i}}. An edge from \textit{v\textsubscript{s}} to \textit{v\textsubscript{d}} with label \textit{l\textsubscript{i}} is denoted by \textit{e}(\textit{v\textsubscript{s}}, \textit{l\textsubscript{i}}, \textit{v\textsubscript{d}}). A path from \textit{v\textsubscript{s}} to \textit{v\textsubscript{d}}, which is a sequence of \textit{e}(\textit{v\textsubscript{s}}, \textit{l\textsubscript{i\textsubscript{1}}}, \textit{v\textsubscript{j\textsubscript{1}}}), \textit{e}(\textit{v\textsubscript{j\textsubscript{1}}}, \textit{l\textsubscript{i\textsubscript{2}}}, \textit{v\textsubscript{j\textsubscript{2}}}), $\cdots$, \textit{e}(\textit{v\textsubscript{j\textsubscript{n-1}}}, \textit{l\textsubscript{i\textsubscript{n}}}, \textit{v\textsubscript{d}}) is denoted by \textit{p}(\textit{v\textsubscript{s}}, \textit{l\textsubscript{i\textsubscript{1}}}, \textit{v\textsubscript{j\textsubscript{1}}}, \textit{l\textsubscript{i\textsubscript{2}}}, \textit{v\textsubscript{j\textsubscript{2}}}, $\cdots$, \textit{l\textsubscript{i\textsubscript{n}}}, \textit{v\textsubscript{d}}). A sequence of labels \textit{l\textsubscript{i\textsubscript{1}}}\textit{l\textsubscript{i\textsubscript{2}}}$\cdots$\textit{l\textsubscript{i\textsubscript{n}}} is called a path label~\cite{Men95,Yak16}. If it is not ambiguous, \textit{p}(\textit{v\textsubscript{s}}, \textit{l\textsubscript{i\textsubscript{1}}}, \textit{v\textsubscript{j\textsubscript{1}}}, \textit{l\textsubscript{i\textsubscript{2}}}, \textit{v\textsubscript{j\textsubscript{2}}}, $\cdots$, \textit{l\textsubscript{i\textsubscript{n}}}, \textit{v\textsubscript{d}}) is simply denoted by \textit{p}(\textit{v\textsubscript{s}}, \textit{v\textsubscript{d}}).

\begin{table}[b]
    \caption{The notation related to \textit{G}.}
    \label{tab:notations for G}
    \centering{}
    \begin{tabular}{|c|l|}
        \hline
        Notation & \multicolumn{1}{|c|}{Description}\\
        \hline
        \hline
        \textit{v\textsubscript{i}} &A vertex with VID i\\
        \hline
        \textit{l\textsubscript{i}} &A label with LID i\\
        \hline
        \textit{e}(\textit{v\textsubscript{s}}, \textit{l\textsubscript{i}}, \textit{v\textsubscript{d}}) &An edge from \textit{v\textsubscript{s}} to \textit{v\textsubscript{d}} with label \textit{l\textsubscript{i}}\\
        \hline
        \begin{tabular}[c]{@{}c@{}} \textit{p}(\textit{v\textsubscript{s}}, \textit{l\textsubscript{i\textsubscript{1}}},$\cdots$, \textit{l\textsubscript{i\textsubscript{n}}}, \textit{v\textsubscript{d}})\\or \textit{p}(\textit{v\textsubscript{s}}, \textit{v\textsubscript{d}})\end{tabular}& \begin{tabular}[c]{@{}l@{}}A path from \textit{v\textsubscript{s}} to \textit{v\textsubscript{d}}, which is a sequence of \\
        e(\textit{v\textsubscript{s}}, \textit{l\textsubscript{i\textsubscript{1}}}, \textit{v\textsubscript{j\textsubscript{1}}}), $\cdots$, e(\textit{v\textsubscript{j\textsubscript{n-1}}}, \textit{l\textsubscript{i\textsubscript{n}}}, \textit{v\textsubscript{d}})\end{tabular}\\
        \hline
        \textit{l\textsubscript{i\textsubscript{1}}}\textit{l\textsubscript{i\textsubscript{2}}}$\cdots$\textit{l\textsubscript{i\textsubscript{n}}} &  \begin{tabular}[c]{@{}l@{}} The path label of a path\\\textit{p}(\textit{v\textsubscript{s}}, \textit{l\textsubscript{i\textsubscript{1}}}, \textit{v\textsubscript{j\textsubscript{1}}}, \textit{l\textsubscript{i\textsubscript{2}}}, \textit{v\textsubscript{j\textsubscript{2}}}, $\cdots$, \textit{l\textsubscript{i\textsubscript{n}}}, \textit{v\textsubscript{d}}) \end{tabular}\\
        \hline
    \end{tabular}
\end{table}

\subsection{Regular Path Query(RPQ)}
\label{subsec:Regular Path Query(RPQ)}
A regular path query (RPQ) \textit{R} on a graph \textit{G} is a regular expression over $\mathit{\Sigma}$ and finds a set of ordered vertex pairs (start vertex, end vertex) of the paths satisfying \textit{R} in \textit{G}~\cite{Men95}~(Definition~\ref{def:path label}). That is, the evaluation result \textit{R\textsubscript{G}} of RPQ \textit{R} on \textit{G} is a set of ordered vertex pairs as is defined in Definition~\ref{def:RPQ result}.

\begin{definition}
\label{def:path label}
A path \textit{satisfies} an RPQ \textit{R} when a path label of the path matches R~\cite{Men95}.
\end{definition}

\begin{definition}
\label{def:RPQ result}
Given a graph \textit{G} and an RPQ \textit{R}, \\\textit{R\textsubscript{G}} = \{{\normalfont (}\textit{v\textsubscript{i}}, \textit{v\textsubscript{j}}{\normalfont )} $\mid$ a p{\normalfont (}\textit{v\textsubscript{i}}, \textit{v\textsubscript{j}}{\normalfont )} that satisfies \textit{R} exists in \textit{G}\}.
\end{definition}

\begin{example}
Fig.~\ref{fig:an example of RPQ} shows an example of \textit{R\textsubscript{G}}. (\textit{d}$\cdot$(\textit{b$\cdot$c})\textsuperscript{+}$\cdot$\textit{c})\textsubscript{\textit{G}} is the result of RPQ \textit{d}$\cdot$(\textit{b$\cdot$c})\textsuperscript{+}$\cdot$\textit{c} on \textit{G} where \textit{G} is an example graph shown in Fig.~\ref{fig:an example of graph}. The left side of the figure shows the paths that satisfy \textit{d}$\cdot$(\textit{b$\cdot$c})\textsuperscript{+}$\cdot$\textit{c}. The path labels of these paths are \textit{dbcc}, \textit{dbcbcc}, etc., all of which match \textit{d}$\cdot$(\textit{b$\cdot$c})\textsuperscript{+}$\cdot$\textit{c}. The right side shows (\textit{d}$\cdot$(\textit{b$\cdot$c})\textsuperscript{+}$\cdot$\textit{c})\textsubscript{\textit{G}}. (\textit{v\textsubscript{7}}, \textit{v\textsubscript{5}}) is included in (\textit{d}$\cdot$(\textit{b$\cdot$c})\textsuperscript{+}$\cdot$\textit{c})\textsubscript{\textit{G}} because paths \textit{p\textsubscript{1}}, \textit{p\textsubscript{3}}, etc., satisfying \textit{d}$\cdot$(\textit{b$\cdot$c})\textsuperscript{+}$\cdot$\textit{c} exist between vertices \textit{v\textsubscript{7}} and \textit{v\textsubscript{5}}. (\textit{v\textsubscript{7}}, \textit{v\textsubscript{3}}) is also included in (\textit{d}$\cdot$(\textit{b$\cdot$c})\textsuperscript{+}$\cdot$\textit{c})\textsubscript{\textit{G}} because paths \textit{p\textsubscript{2}}, \textit{p\textsubscript{4}}, etc., satisfying \textit{d}$\cdot$(\textit{b$\cdot$c})\textsuperscript{+}$\cdot$\textit{c}, exist between vertices \textit{v\textsubscript{7}} and \textit{v\textsubscript{3}}.
\end{example}

\begin{figure}[b]
\centering
\includegraphics[scale= 0.88]{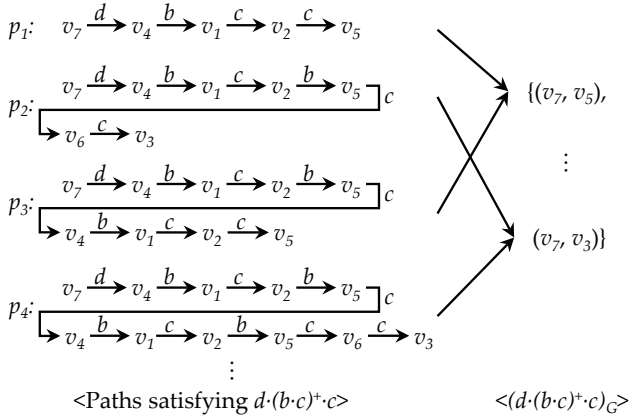}
\vspace{-0.2cm}
\caption{An example of \textit{R\textsubscript{G}}.}
\label{fig:an example of RPQ}
\end{figure}

To evaluate an RPQ \textit{R} on \textit{G}, we traverse \textit{G} from each vertex \textit{v\textsubscript{i}} and find the set of end vertices of the paths starting from \textit{v\textsubscript{i}} that satisfy \textit{R}. To find the paths satisfying \textit{R}, pattern matching is done on the labels of the edges that are accessed during traversal. Finite automata are usually used for pattern matching~\cite{Kos12,Men95,Ngu17,Yak15,Yak16}. Each traversal has a state of the finite automaton as the current state. If the state can be transited to the next state using the label of an edge originating from the current state, the traversal continues through that edge. If the current state is an accept state, (start vertex, end vertex) of the traversal is included in \textit{R\textsubscript{G}}. The traversal continues until there are no more traversable edges or the state is not able to be transited. Among algorithms based on finite automata, we use recent used one by Yakovets~\cite{Yak16} in experiments in Section~\ref{sec:Performance Evaluation}.

\begin{example}
Fig.~\ref{fig:an Example of RPQ Evaluation} shows an example of RPQ evaluation --- evaluation of the RPQ \textit{d}$\cdot$(\textit{b$\cdot$c})\textsuperscript{+}$\cdot$\textit{c} using the finite automata (i.e., NFA in Fig.~\ref{fig:an Example of RPQ Evaluation}) on the graph of Fig.~\ref{fig:an example of graph}. Through the traversal, (\textit{v\textsubscript{7}}, \textit{v\textsubscript{5}}) and (\textit{v\textsubscript{7}}, \textit{v\textsubscript{3}}) are included in (\textit{d}$\cdot$(\textit{b$\cdot$c})\textsuperscript{+}$\cdot$\textit{c})\textsubscript{\textit{G}} since \textit{p}(\textit{v\textsubscript{7}}, \textit{d}, \textit{v\textsubscript{4}}, \textit{b}, \textit{v\textsubscript{1}}, \textit{c}, \textit{v\textsubscript{2}}, \textit{c}, \textit{v\textsubscript{5}}) and \textit{p}(\textit{v\textsubscript{7}}, \textit{d}, \textit{v\textsubscript{4}}, \textit{b}, \textit{v\textsubscript{1}}, \textit{c}, \textit{v\textsubscript{2}}, \textit{b}, \textit{v\textsubscript{5}}, \textit{c}, \textit{v\textsubscript{6}}, \textit{c}, \textit{v\textsubscript{3}}), which satisfies \textit{d}$\cdot$(\textit{b$\cdot$c})\textsuperscript{+}$\cdot$\textit{c} starting from \textit{v\textsubscript{7}}, are found. In the case of \textit{p}(\textit{v\textsubscript{7}}, \textit{d}, \textit{v\textsubscript{4}}, \textit{b}, \textit{v\textsubscript{1}}, \textit{c}, \textit{v\textsubscript{2}}, \textit{b}, \textit{v\textsubscript{3}}), there is an edge \textit{e}(\textit{v\textsubscript{3}}, \textit{b}, \textit{v\textsubscript{2}}) accessible from \textit{v\textsubscript{3}}, but state transition cannot be done through the edge. Therefore, the traversal of the path is terminated. \textit{p}(\textit{v\textsubscript{7}}, \textit{d}, \textit{v\textsubscript{4}}, \textit{b}, \textit{v\textsubscript{1}}, \textit{c}, \textit{v\textsubscript{2}}, \textit{b}, \textit{v\textsubscript{5}}, \textit{c}, \textit{v\textsubscript{4}}, \textit{b}, \textit{v\textsubscript{1}}) is a special case. Since the end vertex \textit{v\textsubscript{1}} of the path has already been visited by \textit{p}(\textit{v\textsubscript{7}}, \textit{d}, \textit{v\textsubscript{4}}, \textit{b}, \textit{v\textsubscript{1}}) in the same state (i.e., \textit{q\textsubscript{2}} in Fig.~\ref{fig:an Example of RPQ Evaluation}), the subsequent traversals overlap with an earlier one. That is, the vertices that are included in (\textit{d}$\cdot$(\textit{b$\cdot$c})\textsuperscript{+}$\cdot$\textit{c})\textsubscript{\textit{G}} by paths from \textit{v\textsubscript{7}} found through subsequent traversals are duplicated. Therefore, if the end vertex of the path has already been visited in the same state from the start vertex of the path, the traversal is terminated to avoid duplication.
\end{example}

\begin{figure}[t]
\centering
\includegraphics[scale= 0.88]{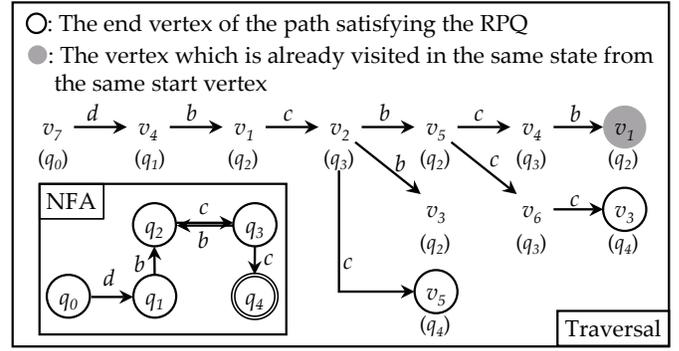}
\caption{An example of RPQ evaluation.}
\label{fig:an Example of RPQ Evaluation}
\end{figure}
 
As we can see in Example 2, evaluation of RPQ \textit{R} on \textit{G} is a complex operation since it requires both graph traversal and pattern matching~\cite{Gra03a}. In particular, evaluation of an RPQ with Kleene closures is more expensive and the result is larger than those of RPQs without them because the former can have long paths of varied lengths that satisfy the RPQ~\cite{Gra03b}. Time and space complexities of the naive RPQ evaluation depends on those of Kleene closure evaluation, which are O($\mid$\textit{V}$\mid$$\times$$\mid$\textit{E}$\mid$) and O($\mid$\textit{V}$\mid$\textsuperscript{\textit{2}}), respectively. Therefore, optimization of the evaluation of an RPQ with Kleene closures is a very important issue. 

\subsection{Evaluation of Multiple RPQs}
\label{subsec:Evaluation Method of Multiple RPQs}
A naive method of evaluating multiple RPQs is to evaluate them individually. However, if there is a common sub-query among RPQs, the sub-query is evaluated repeatedly. To avoid the problem, Abul-Basher's method~\cite{Abu17} shares the evaluation result of the common sub-query among the RPQs. We call the method \textit{FullSharing} and use it as a baseline solution in the experiment. However, if a common sub-query is a Kleene closure, the one-time evaluation of the common sub-query itself is expensive and the size of the result could be large as mentioned in Section~\ref{subsec:Regular Path Query(RPQ)}. In addition, \textit{useless} and \textit{redundant operations} can occur when each RPQ is evaluated using the result of the common sub-query. We discuss the details of these problems and propose solutions in Section~\ref{sec:Multiple Regular Path Queries(MRPQs) Evaluation Method}.

\section{RPQ-based graph reduction}
\label{sec:Intermediate Representation}
In this section, we present a novel concept of \textit{RPQ-based graph reduction}. Fig.~\ref{fig:an overview of the graph reduction} shows an overview of RPQ-based graph reduction. The first-level reduction (\textit{G} $\rightarrow$ \textit{G\textsubscript{R}} of Section~\ref{subsubsec:Edge-Level Reduced Graph}) reduces a graph \textit{G} at the edge level for RPQ \textit{R}, which maps the paths satisfying \textit{R} on \textit{G} to edges on \textit{G\textsubscript{R}}. The second-level reduction (\textit{G\textsubscript{R}} $\rightarrow$ \textit{\textoverline{G\textsubscript{R}}} of Section~\ref{subsubsec:Vertex-Level Reduced Graph}) reduces \textit{G\textsubscript{R}} at the vertex level, which maps each SCC of \textit{G\textsubscript{R}} to a vertex of \textit{\textoverline{G\textsubscript{R}}}. Based on these graph reductions, we show that \textit{R\textsuperscript{+}} can be evaluated as the transitive closure of the edge-level reduced graph \textit{G\textsubscript{R}} (Lemma~\ref{lemma:RstarG equals RTC of GR}), and further, can be calculated as the transitive closure of the two-level reduced graph \textit{\textoverline{G\textsubscript{R}}} and the Cartesian product of vertices of SCCs mapped to vertices of \textit{\textoverline{G\textsubscript{R}}} (Lemma~\ref{lemma:RTC of GR equals Cartesian product of RTC of GRbar} and Theorem~\ref{theorem:R+G equals Cartesian product of RTC of GRbar}). We also propose the transitive closure of \textit{\textoverline{G\textsubscript{R}}} as a lightweight structure for efficiently sharing the result of \textit{R\textsuperscript{+}}, which we call a reduced transitive closure (RTC of Section~\ref{subsec:Intermediate Representation of RstarG}).

\begin{figure}[b]
\centering
\includegraphics[scale=0.82]{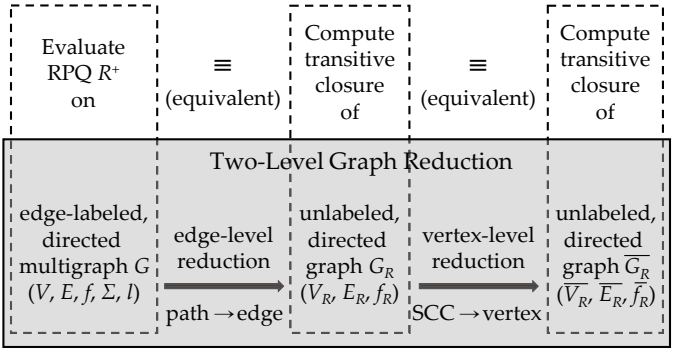}
\caption{An overview of RPQ-based graph reduction.}
\label{fig:an overview of the graph reduction}
\end{figure}

TABLE~\ref{tab:notations for GR} summarizes the notation related to \textit{G\textsubscript{R}} and \textit{\textoverline{G\textsubscript{R}}}. Vertex, edge, and path notations of \textit{G\textsubscript{R}} are the same as those of \textit{G} in TABLE~\ref{tab:notations for G} except that we use \textit{e\textsubscript{R}} and \textit{p\textsubscript{R}} instead of \textit{e} and \textit{p} to distinguish \textit{G\textsubscript{R}} and \textit{G}. Each SCC of \textit{G\textsubscript{R}} has a unique ID (SID). An SCC with SID \textit{i} is denoted by \textit{s\textsubscript{i}}. If it is not ambiguous, the set of vertices in the SCC with SID \textit{i} is also denoted by \textit{s\textsubscript{i}}. Each vertex of \textit{\textoverline{G\textsubscript{R}}} has the same unique ID as that of the SCC of \textit{G\textsubscript{R}} mapped to it. A vertex with VID \textit{i} is denoted by \textit{\textoverline{v\textsubscript{i}}}. An edge from \textit{\textoverline{v\textsubscript{s}}} to \textit{\textoverline{v\textsubscript{d}}} is denoted by \textit{\textoverline{e\textsubscript{R}}}(\textit{\textoverline{v\textsubscript{s}}}, \textit{\textoverline{v\textsubscript{d}}}).

\subsection{Edge-Level Graph Reduction (G $\rightarrow$ \textit{G\textsubscript{R}})}
\label{subsubsec:Edge-Level Reduced Graph}
The edge-level reduction maps all the paths satisfying \textit{R} between each pair of vertices of \textit{G} to one edge of \textit{G\textsubscript{R}}. Graph \textit{G\textsubscript{R}} is an unlabeled, directed graph that reduces \textit{G} at the edge level for \textit{R}. \textit{G\textsubscript{R}} is defined as a 3-tuple (\textit{V\textsubscript{R}}, \textit{E\textsubscript{R}}, \textit{f\textsubscript{R}}). \textit{V\textsubscript{R}} is a set of vertices: \{\textit{v\textsubscript{i}} $\mid$ ($\exists$\textit{v\textsubscript{j}})((\textit{v\textsubscript{i}}, \textit{v\textsubscript{j}}) $\in$ \textit{E\textsubscript{R}} $\vee$ (\textit{v\textsubscript{j}}, \textit{v\textsubscript{i}}) $\in$ \textit{E\textsubscript{R}})\}. \textit{E\textsubscript{R}} is a set of edges: \{\textit{e\textsubscript{R}}(\textit{v\textsubscript{i}}, \textit{v\textsubscript{j}}) $\mid$ there exists \textit{p}(\textit{v\textsubscript{i}}, \textit{v\textsubscript{j}}) satisfying \textit{R} on \textit{G}\}. \textit{f\textsubscript{R}}: \textit{E\textsubscript{R}} $\rightarrow$ \textit{V\textsubscript{R}}$\times$\textit{V\textsubscript{R}} is a function that maps each edge to an ordered pair of vertices connected by the edge. Reduction of \textit{G} to \textit{G\textsubscript{R}} through the edge-level reduction has the following aspects.

\begin{table}[t]
    \caption{The notation related to \textit{G\textsubscript{R}} and \textit{\textoverline{G\textsubscript{R}}}.}
    \label{tab:notations for GR}
    \centering{}
    \begin{tabular}{|c|c||p{5.8cm}|}
        \hline
        \multicolumn{2}{|c||}{Notation} & \multicolumn{1}{c|}{Description}\\
        \hline
        \hline
        \multirow{6}{*}{\textit{G\textsubscript{R}}}& \textit{v\textsubscript{i}} & A vertex of \textit{G\textsubscript{R}} with VID $i$\\
        \cline{2-3}
        & {\textit{e\textsubscript{R}}(\textit{v\textsubscript{s}}, \textit{v\textsubscript{d}})} & An edge of \textit{G\textsubscript{R}} from \textit{v\textsubscript{s}} to \textit{v\textsubscript{d}}\\
        \cline{2-3}
        & \multirow{2}{*}{\textit{p\textsubscript{R}}(\textit{v\textsubscript{s}}, \textit{v\textsubscript{d}})} & A path of \textit{G\textsubscript{R}} from \textit{v\textsubscript{s}} to \textit{v\textsubscript{d}}, which is a sequence of \textit{e\textsubscript{R}}(\textit{v\textsubscript{s}}, \textit{v\textsubscript{j\textsubscript{1}}}), $\cdots$, \textit{e\textsubscript{R}}(\textit{v\textsubscript{j\textsubscript{n-1}}},\textit{v\textsubscript{d}})\\
        \cline{2-3}
        & \multirow{2}{*}{\textit{s\textsubscript{i}}} &An SCC of \textit{G\textsubscript{R}} with SID $i$ or the set of vertices in the SCC\\
        \hline
        \multirow{3}{*}{\textit{\textoverline{G\textsubscript{R}}}} & \multirow{2}{*}{\textit{\textoverline{v\textsubscript{i}}}}& The vertex of \textit{\textoverline{G\textsubscript{R}}} to which \textit{s\textsubscript{i}} of \textit{G\textsubscript{R}} is mapped (VID is \textit{i}) \\
        \cline{2-3}
        & \textit{\textoverline{e\textsubscript{R}}}(\textit{\textoverline{v\textsubscript{s}}}, \textit{\textoverline{v\textsubscript{d}}}) & An edge of \textit{\textoverline{G\textsubscript{R}}} from \textit{\textoverline{v\textsubscript{s}}} to \textit{\textoverline{v\textsubscript{d}}}\\
        \hline
    \end{tabular}
\end{table}

\begin{itemize}
    \item Vertices and edges that do not belong to a path satisfying \textit{R} are excluded. To obtain \textit{R\textsuperscript{+}\hspace{-0.12cm}\textsubscript{G}} we only need the paths satisfying \textit{R}. Thus, we don't need to consider vertices and edges not belonging to the paths satisfying \textit{R} on \textit{G}.
    \item Labeled graph $\rightarrow$ unlabeled graph. Since only the paths satisfying \textit{R} are mapped to edges, labels of edges are no longer needed. (i.e., every label is \textit{R}.)
    \item Multigraph $\rightarrow$ simple graph. Since labels of edges are excluded, paths in the same direction between each vertex pair are mapped to one edge.
\end{itemize}

\begin{example}
Fig.~\ref{fig:an example of edge level graph reduction} shows an example of edge-level graph reduction. Here, the graph \textit{G} in Fig.~\ref{fig:an example of graph} is reduced at the edge level for \textit{b$\cdot$c}. In \textit{G}, the paths satisfying \textit{b$\cdot$c} are \textit{p}(\textit{v\textsubscript{2}}, \textit{v\textsubscript{4}}), \textit{p}(\textit{v\textsubscript{2}}, \textit{v\textsubscript{6}}), \textit{p}(\textit{v\textsubscript{3}}, \textit{v\textsubscript{5}}), \textit{p}(\textit{v\textsubscript{4}}, \textit{v\textsubscript{2}}), and \textit{p}(\textit{v\textsubscript{5}}, \textit{v\textsubscript{3}}). Since these paths are mapped to edges of \textit{G\textsubscript{b$\cdot$c}}, \textit{E\textsubscript{b$\cdot$c}} becomes \{\textit{e\textsubscript{b$\cdot$c}}(\textit{v\textsubscript{2}}, \textit{v\textsubscript{4}}), \textit{e\textsubscript{b$\cdot$c}}(\textit{v\textsubscript{2}}, \textit{v\textsubscript{6}}), \textit{e\textsubscript{b$\cdot$c}}(\textit{v\textsubscript{3}}, \textit{v\textsubscript{5}}), \textit{e\textsubscript{b$\cdot$c}}(\textit{v\textsubscript{4}}, \textit{v\textsubscript{2}}), \textit{e\textsubscript{b$\cdot$c}}(\textit{v\textsubscript{5}}, \textit{v\textsubscript{3}})\}.
\end{example}

\begin{figure}[b]
\centering
\includegraphics[scale=0.88]{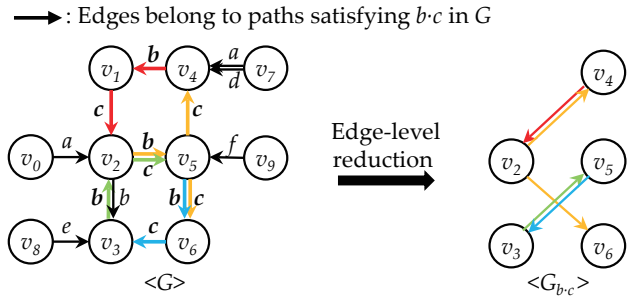}
\caption{An example of edge-level reduction.}
\label{fig:an example of edge level graph reduction}
\end{figure}

We show that the problem of evaluating RPQ \textit{R\textsuperscript{+}} can be reduced to a problem of computing the transitive closure of the edge-level reduced graph \textit{G\textsubscript{R}} (denoted by TC(\textit{G\textsubscript{R}})) in Lemma~\ref{lemma:RstarG equals RTC of GR}.

\begin{lemma}
\label{lemma:RstarG equals RTC of GR}
\textit{R\textsuperscript{+}\hspace{-0.12cm}\textsubscript{G}} is equivalent to TC(\textit{G\textsubscript{R}}).
\end{lemma}

\begin{proof}
We prove \textit{R\textsuperscript{+}\hspace{-0.12cm}\textsubscript{G}} = TC(\textit{G\textsubscript{R}}) by showing \textit{R\textsuperscript{+}\hspace{-0.12cm}\textsubscript{G}} $\subseteq$ TC(\textit{G\textsubscript{R}}) and \textit{R\textsuperscript{+}\hspace{-0.12cm}\textsubscript{G}} $\supseteq$ TC(\textit{G\textsubscript{R}}), respectively.

\begin{itemize}
    \item \textit{R\textsuperscript{+}\hspace{-0.12cm}\textsubscript{G}} $\subseteq$ TC(\textit{G\textsubscript{R}})\\
    Suppose (\textit{v\textsubscript{i}}, \textit{v\textsubscript{j}}) $\in$ \textit{R\textsuperscript{+}\hspace{-0.12cm}\textsubscript{G}}. Then, there exists a path \textit{p}(\textit{v\textsubscript{i}}, \textit{v\textsubscript{j}}) satisfying \textit{R\textsuperscript{n}} (concatenation of \textit{R} \textit{n} times, \textit{n} $\ge$ 1) on \textit{G}. \textit{p}(\textit{v\textsubscript{i}}, \textit{v\textsubscript{j}}) can be represented as a sequence of \textit{n} paths \textit{p}(\textit{v\textsubscript{i}}, \textit{v\textsubscript{k\textsubscript{1}}}), \textit{p}(\textit{v\textsubscript{k\textsubscript{1}}}, \textit{v\textsubscript{k\textsubscript{2}}}), $\cdots$, \textit{p}(\textit{v\textsubscript{k\textsubscript{n-1}}}, \textit{v\textsubscript{j}}) each satisfying \textit{R}. By edge-level reduction, all the paths satisfying \textit{R} between each pair of vertices maps to one edge in \textit{G\textsubscript{R}}. That is, because there exist \textit{e\textsubscript{R}}(\textit{v\textsubscript{i}}, \textit{v\textsubscript{k\textsubscript{1}}}), \textit{e\textsubscript{R}}(\textit{v\textsubscript{k\textsubscript{1}}}, \textit{v\textsubscript{k\textsubscript{2}}}), $\cdots$, \textit{e\textsubscript{R}}(\textit{v\textsubscript{k\textsubscript{n-1}}}, \textit{v\textsubscript{j}}) in \textit{G\textsubscript{R}}, there exist \textit{p\textsubscript{R}}(\textit{v\textsubscript{i}}, \textit{v\textsubscript{j}}) in \textit{G\textsubscript{R}}. Hence, (\textit{v\textsubscript{i}}, \textit{v\textsubscript{j}}) $\in$ TC(\textit{G\textsubscript{R}}) and \textit{R\textsuperscript{+}\hspace{-0.12cm}\textsubscript{G}} $\subseteq$ TC(\textit{G\textsubscript{R}}).

    \item \textit{R\textsuperscript{+}\hspace{-0.12cm}\textsubscript{G}} $\supseteq$ TC(\textit{G\textsubscript{R}})\\
    Suppose (\textit{v\textsubscript{i}}, \textit{v\textsubscript{j}}) $\in$ TC(\textit{G\textsubscript{R}}). Then, there exists a path \textit{p\textsubscript{R}}(\textit{v\textsubscript{i}}, \textit{v\textsubscript{j}}) of length \textit{n}(\textit{n} $\ge$ 1) on \textit{G\textsubscript{R}}. \textit{p\textsubscript{R}}(\textit{v\textsubscript{i}}, \textit{v\textsubscript{j}}) can be represented as a sequence of \textit{n} edges \textit{e\textsubscript{R}}(\textit{v\textsubscript{i}}, \textit{v\textsubscript{k\textsubscript{1}}}), \textit{e\textsubscript{R}}(\textit{v\textsubscript{k\textsubscript{1}}}, \textit{v\textsubscript{k\textsubscript{2}}}), $\cdots$, \textit{e\textsubscript{R}}(\textit{v\textsubscript{k\textsubscript{n-1}}}, \textit{v\textsubscript{j}}). By the definition of \textit{G\textsubscript{R}}, there exist paths \textit{p}(\textit{v\textsubscript{i}}, \textit{v\textsubscript{k\textsubscript{1}}}), \textit{p}(\textit{v\textsubscript{k\textsubscript{1}}}, \textit{v\textsubscript{k\textsubscript{2}}}), $\cdots$, \textit{p}(\textit{v\textsubscript{k\textsubscript{n-1}}}, \textit{v\textsubscript{j}}) each satisfying \textit{R} in \textit{G}. That is, there exists a path \textit{p}(\textit{v\textsubscript{i}}, \textit{v\textsubscript{j}}) satisfying \textit{R\textsuperscript{n}}. Hence, (\textit{v\textsubscript{i}}, \textit{v\textsubscript{j}}) $\in$ \textit{R\textsuperscript{+}\hspace{-0.12cm}\textsubscript{G}} and \textit{R\textsuperscript{+}\hspace{-0.12cm}\textsubscript{G}} $\supseteq$ TC(\textit{G\textsubscript{R}}).
\end{itemize}
\vspace{-0.2cm}
\end{proof}

\begin{example}
In Fig.~\ref{fig:an example of edge level graph reduction}, (\textit{b$\cdot$c})\textit{\textsuperscript{+}\hspace{-0.12cm}\textsubscript{G}} is equivalent to TC(\textit{G}\textsubscript{\textit{b$\cdot$c}}):\\\{(\textit{v\textsubscript{2}}, \textit{v\textsubscript{2}}), (\textit{v\textsubscript{2}}, \textit{v\textsubscript{4}}), (\textit{v\textsubscript{2}}, \textit{v\textsubscript{6}}), (\textit{v\textsubscript{3}}, \textit{v\textsubscript{3}}), (\textit{v\textsubscript{3}}, \textit{v\textsubscript{5}}), (\textit{v\textsubscript{4}}, \textit{v\textsubscript{2}}), (\textit{v\textsubscript{4}}, \textit{v\textsubscript{4}}), (\textit{v\textsubscript{4}}, \textit{v\textsubscript{6}}), (\textit{v\textsubscript{5}}, \textit{v\textsubscript{3}}), (\textit{v\textsubscript{5}}, \textit{v\textsubscript{5}})\}.
\end{example}

\subsection{Vertex-Level Graph Reduction (\textit{G\textsubscript{R}} $\rightarrow$ \textit{\textoverline{G\textsubscript{R}}})}
\label{subsubsec:Vertex-Level Reduced Graph}
The vertex-level reduction maps each SCC of \textit{G\textsubscript{R}} to one vertex of \textit{\textoverline{G\textsubscript{R}}}. The graph \textit{\textoverline{G\textsubscript{R}}} is an unlabeled, directed graph. \textit{\textoverline{G\textsubscript{R}}} is defined as a 3-tuple (\textit{\textoverline{V\textsubscript{R}}}, \textit{\textoverline{E\textsubscript{R}}}, \textit{\textoverline{f\textsubscript{R}}}). \textit{\textoverline{V\textsubscript{R}}} is a set of vertices: \{\textit{\textoverline{v\textsubscript{i}}} $\mid$ ($\exists$\textit{s\textsubscript{j}})(\textit{s\textsubscript{j}} is an SCC of \textit{G\textsubscript{R}} $\wedge$ \textit{i} = \textit{j})\}. \textit{\textoverline{E\textsubscript{R}}} is a set of edges: \{\textit{\textoverline{e\textsubscript{R}}}(\textit{\textoverline{v\textsubscript{i}}}, \textit{\textoverline{v\textsubscript{j}}}) $\mid$ ($\exists$ \textit{s\textsubscript{k}})($\exists$ \textit{s\textsubscript{l}})($\exists$ \textit{v\textsubscript{m}})($\exists$ \textit{v\textsubscript{n}})(\textit{s\textsubscript{k}} and \textit{s\textsubscript{l}} are SCCs of \textit{G\textsubscript{R}} $\wedge$ \textit{v\textsubscript{m}} $\in$ \textit{s\textsubscript{k}} $\wedge$ \textit{v\textsubscript{n}} $\in$ \textit{s\textsubscript{l}} $\wedge$ \textit{e\textsubscript{R}}(\textit{v\textsubscript{m}}, \textit{v\textsubscript{n}}) $\in$ \textit{E\textsubscript{R}} $\wedge$ \textit{i} = \textit{k} $\wedge$ \textit{j} = \textit{l})\}. \textit{\textoverline{f\textsubscript{R}}}: \textit{\textoverline{E\textsubscript{R}}} $\rightarrow$ \textit{\textoverline{V\textsubscript{R}}}$\times$\textit{\textoverline{V\textsubscript{R}}} is a function that maps each edge to an ordered pair of vertices connected by the edge. The vertex-level reduction has the following characteristics:

\begin{itemize}
    \item Edges between any pair of vertices in the same SCC of \textit{G\textsubscript{R}} are mapped to one self-loop edge in \textit{\textoverline{G\textsubscript{R}}}.
    \item Edges with the same direction between any pair of vertices in two different SCCs of \textit{G\textsubscript{R}} are mapped to one edge in \textit{\textoverline{G\textsubscript{R}}}.
\end{itemize}

\begin{figure}[b]
\centering
\includegraphics[scale=0.88]{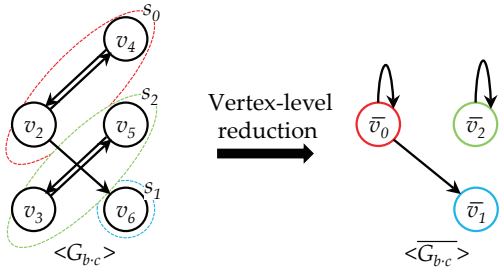}
\caption{An example of vertex-level reduction.}
\label{fig:an example of vertex level graph reduction}
\end{figure}

\begin{example}
Fig.~\ref{fig:an example of vertex level graph reduction} shows an example of vertex-level graph reduction. Here, \textit{G}\textsubscript{\textit{b$\cdot$c}} in Fig.~\ref{fig:an example of edge level graph reduction} is reduced to \textit{\textoverline{G\textsubscript{b$\cdot$c}}} at the vertex level. There are three SCCs in \textit{G}\textsubscript{\textit{b$\cdot$c}}: \textit{s\textsubscript{0}}, \textit{s\textsubscript{1}}, \textit{s\textsubscript{2}}. Since each of these SCCs is mapped to one vertex of \textit{\textoverline{G\textsubscript{b$\cdot$c}}}, \textit{\textoverline{V\textsubscript{b$\cdot$c}}} becomes \{\textit{\textoverline{v\textsubscript{0}}}, \textit{\textoverline{v\textsubscript{1}}}, \textit{\textoverline{v\textsubscript{2}}}\}. There exists an edge from \textit{v\textsubscript{2}} in \textit{s\textsubscript{0}} to \textit{v\textsubscript{6}} in \textit{s\textsubscript{1}}; these vertices belong to different SCCs. This edge of \textit{G\textsubscript{R}} is mapped to the edge \textit{\textoverline{e\textsubscript{b$\cdot$c}}}(\textit{\textoverline{v\textsubscript{0}}}, \textit{\textoverline{v\textsubscript{1}}}) of \textit{\textoverline{G\textsubscript{R}}}. The SCC \textit{s\textsubscript{0}} (\textit{s\textsubscript{2}}) has edges between the vertices that belong to \textit{s\textsubscript{0}} (\textit{s\textsubscript{2}}) itself. These edges are mapped to the edge \textit{\textoverline{e\textsubscript{b$\cdot$c}}}(\textit{\textoverline{v\textsubscript{0}}}, \textit{\textoverline{v\textsubscript{0}}}) (\textit{\textoverline{e\textsubscript{b$\cdot$c}}}(\textit{\textoverline{v\textsubscript{2}}}, \textit{\textoverline{v\textsubscript{2}}})) in \textit{\textoverline{G\textsubscript{R}}} constituting an self-loop edge. Thus, \textit{\textoverline{E\textsubscript{b$\cdot$c}}} is \{\textit{\textoverline{e\textsubscript{b$\cdot$c}}}(\textit{\textoverline{v\textsubscript{0}}}, \textit{\textoverline{v\textsubscript{0}}}), \textit{\textoverline{e\textsubscript{b$\cdot$c}}}(\textit{\textoverline{v\textsubscript{0}}}, \textit{\textoverline{v\textsubscript{1}}}), \textit{\textoverline{e\textsubscript{b$\cdot$c}}}(\textit{\textoverline{v\textsubscript{2}}}, \textit{\textoverline{v\textsubscript{2}}})\}.
\end{example}

Using Lemma~\ref{lemma:RTC of GR equals Cartesian product of RTC of GRbar} below, we can efficiently compute the TC(\textit{G\textsubscript{R}}) (and accordingly, \textit{R\textsuperscript{+}\hspace{-0.12cm}\textsubscript{G}}) by computing the transitive closure of \textit{\textoverline{G\textsubscript{R}}} (denoted by TC(\textit{\textoverline{G\textsubscript{R}}})). Purdom~\cite{Pur70} provided the following Lemma~\ref{lemma: the characteristic of SCC}.

\begin{lemma}[\cite{Pur70}]
\label{lemma: the characteristic of SCC}
For any pair of vertices \textit{v\textsubscript{i}} and \textit{v\textsubscript{j}}, if there is a path from \textit{v\textsubscript{i}} to \textit{v\textsubscript{j}}, there are paths from each vertex belonging to the same SCC as that of \textit{v\textsubscript{i}} to all vertices belonging to the same SCC as that of \textit{v\textsubscript{j}}.
\end{lemma}

We now introduce the following Lemma~\ref{lemma:RTC of GR equals Cartesian product of RTC of GRbar}.

\begin{lemma}
\label{lemma:RTC of GR equals Cartesian product of RTC of GRbar}
TC(\textit{G\textsubscript{R}}) = \{(\textit{v\textsubscript{i}}, \textit{v\textsubscript{j}}) $\mid$ (\textit{\textoverline{v\textsubscript{k}}}, \textit{\textoverline{v\textsubscript{l}}}) $\in$ TC(\textit{\textoverline{G\textsubscript{R}}}) $\wedge$ (\textit{v\textsubscript{i}}, \textit{v\textsubscript{j}}) $\in$ \textit{s\textsubscript{k}}$\times$\textit{s\textsubscript{l}}\}, where \textit{\textoverline{v\textsubscript{k}}}, \textit{\textoverline{v\textsubscript{l}}} are vertices of \textit{\textoverline{G\textsubscript{R}}} to which \textit{s\textsubscript{k}} and \textit{s\textsubscript{l}} of \textit{G\textsubscript{R}} are mapped, respectively.
\end{lemma}

\begin{proof}
~\begin{itemize}
    \item TC(\textit{G\textsubscript{R}}) $\subseteq$ \{(\textit{v\textsubscript{i}}, \textit{v\textsubscript{j}}) $\mid$ (\textit{\textoverline{v\textsubscript{k}}}, \textit{\textoverline{v\textsubscript{l}}}) $\in$ TC(\textit{\textoverline{G\textsubscript{R}}}) $\wedge$ (\textit{v\textsubscript{i}}, \textit{v\textsubscript{j}}) $\in$ \textit{s\textsubscript{k}}$\times$\textit{s\textsubscript{l}}\}.\\
    Suppose (\textit{v\textsubscript{i}}, \textit{v\textsubscript{j}}) $\in$ TC(\textit{G\textsubscript{R}}). Then, there is a path from a vertex of \textit{\textoverline{G\textsubscript{R}}} which is mapped to the SCC containing \textit{v\textsubscript{i}} to a vertex of \textit{\textoverline{G\textsubscript{R}}} which is mapped to the SCC containing \textit{v\textsubscript{j}} by the definition of \textit{\textoverline{G\textsubscript{R}}}. Hence, (\textit{v\textsubscript{i}}, \textit{v\textsubscript{j}}) $\in$ \{(\textit{v\textsubscript{i}}, \textit{v\textsubscript{j}}) $\mid$ (\textit{\textoverline{v\textsubscript{k}}}, \textit{\textoverline{v\textsubscript{l}}}) $\in$ TC(\textit{\textoverline{G\textsubscript{R}}}) $\wedge$ (\textit{v\textsubscript{i}}, \textit{v\textsubscript{j}}) $\in$ \textit{s\textsubscript{k}}$\times$\textit{s\textsubscript{l}}\} holds.
    
    \item \{(\textit{v\textsubscript{i}}, \textit{v\textsubscript{j}}) $\mid$ (\textit{\textoverline{v\textsubscript{k}}}, \textit{\textoverline{v\textsubscript{l}}}) $\in$ TC(\textit{\textoverline{G\textsubscript{R}}}) $\wedge$ (\textit{v\textsubscript{i}}, \textit{v\textsubscript{j}}) $\in$ \textit{s\textsubscript{k}}$\times$\textit{s\textsubscript{l}}\} $\subseteq$ TC(\textit{G\textsubscript{R}}).\\
    Lemma~\ref{lemma: the characteristic of SCC} means that if there is a path from \textit{v\textsubscript{i}} in \textit{s\textsubscript{k}} to \textit{v\textsubscript{j}} in \textit{s\textsubscript{l}}, then there exist a path from each vertex in \textit{s\textsubscript{k}} to each vertex in \textit{s\textsubscript{l}}. Hence, \{(\textit{v\textsubscript{i}}, \textit{v\textsubscript{j}}) $\mid$ (\textit{\textoverline{v\textsubscript{k}}}, \textit{\textoverline{v\textsubscript{l}}}) $\in$ TC(\textit{\textoverline{G\textsubscript{R}}}) $\wedge$ (\textit{v\textsubscript{i}}, \textit{v\textsubscript{j}}) $\in$ \textit{s\textsubscript{k}}$\times$\textit{s\textsubscript{l}}\} $\subseteq$ TC(\textit{G\textsubscript{R}}) holds by Lemma~\ref{lemma: the characteristic of SCC}.
\end{itemize}
\end{proof}

The transitive closure algorithms in~\cite{Pur70} and \cite{Nuu94} are instances of implementation of Lemma~\ref{lemma:RTC of GR equals Cartesian product of RTC of GRbar} (without a formal correctness proof). Purdom~\cite{Pur70} proposed an algorithm that essentially computes the transitive closure of \textit{\textoverline{G\textsubscript{R}}}, and then, the Cartesian product \textit{s\textsubscript{k}}$\times$\textit{s\textsubscript{l}} without formally introducing the concept of the vertex-reduced graph, but instead, treating the nodes in an SCC of the original graph as an equivalent class. Nuutila~\cite{Nuu94} improved Purdom's algorithm by obtaining the transitive closure of \textit{\textoverline{G\textsubscript{R}}} and the Cartesian product \textit{s\textsubscript{k}}$\times$\textit{s\textsubscript{l}} in one step in an interleaved way. In this paper, we formalize the concept implied by the algorithm by Prudom~\cite{Pur70} with a formal definition of the vertex reduced graph and the correctness proof.

\begin{theorem}
\label{theorem:R+G equals Cartesian product of RTC of GRbar}
\textit{R\textsuperscript{+}\hspace{-0.12cm}\textsubscript{G}} = \{(\textit{v\textsubscript{i}}, \textit{v\textsubscript{j}}) $\mid$ (\textit{\textoverline{v\textsubscript{k}}}, \textit{\textoverline{v\textsubscript{l}}}) $\in$ TC(\textit{\textoverline{G\textsubscript{R}}}) $\wedge$ (\textit{v\textsubscript{i}}, \textit{v\textsubscript{j}}) $\in$ \textit{s\textsubscript{k}}$\times$\textit{s\textsubscript{l}}\} where \textit{\textoverline{v\textsubscript{k}}}, \textit{\textoverline{v\textsubscript{l}}} are vertices of \textit{\textoverline{G\textsubscript{R}}} to which \textit{s\textsubscript{k}} and \textit{s\textsubscript{l}} of \textit{G\textsubscript{R}} are mapped, respectively.
\end{theorem}
\begin{proof}
Derived from Lemmas~\ref{lemma:RstarG equals RTC of GR} and \ref{lemma:RTC of GR equals Cartesian product of RTC of GRbar}.
\end{proof}

\begin{example}
In Fig.~\ref{fig:an example of vertex level graph reduction}, TC(\textit{\textoverline{G\textsubscript{b$\cdot$c}}}) is \{(\textit{\textoverline{v\textsubscript{0}}},\textit{\textoverline{v\textsubscript{0}}}), (\textit{\textoverline{v\textsubscript{0}}},\textit{\textoverline{v\textsubscript{1}}}), (\textit{\textoverline{v\textsubscript{2}}},\textit{\textoverline{v\textsubscript{2}}})\}. For each vertex pair (\textit{\textoverline{v\textsubscript{i}}}, \textit{\textoverline{v\textsubscript{j}}}), the union of the Cartesian product of \textit{s\textsubscript{i}} and \textit{s\textsubscript{j}} is \{(\textit{v\textsubscript{2}}, \textit{v\textsubscript{2}}), (\textit{v\textsubscript{2}}, \textit{v\textsubscript{4}}), (\textit{v\textsubscript{4}}, \textit{v\textsubscript{4}}), (\textit{v\textsubscript{4}}, \textit{v\textsubscript{2}}), (\textit{v\textsubscript{2}}, \textit{v\textsubscript{6}}), (\textit{v\textsubscript{4}}, \textit{v\textsubscript{6}}), (\textit{v\textsubscript{3}}, \textit{v\textsubscript{3}}), (\textit{v\textsubscript{3}}, \textit{v\textsubscript{5}}), (\textit{v\textsubscript{5}}, \textit{v\textsubscript{3}}), (\textit{v\textsubscript{5}}, \textit{v\textsubscript{5}})\}, which is the same as TC(\textit{G}\textsubscript{\textit{b$\cdot$c}}). 
\end{example}

\subsection{Reduced Transitive Closure (RTC)} 
\label{subsec:Intermediate Representation of RstarG}
We use TC(\textit{\textoverline{G\textsubscript{R}}}) (denoted by \textit{\textoverline{R\textsuperscript{+}\hspace{-0.12cm}\textsubscript{G}}}) as a reduced transitive closure to share the result of a Kleene plus \textit{R\textsuperscript{+}}. As illustrated in Section~\ref{subsubsec:Vertex-Level Reduced Graph}, we can efficiently enumerate \textit{R\textsuperscript{+}\hspace{-0.12cm}\textsubscript{G}} by using \textit{\textoverline{R\textsuperscript{+}\hspace{-0.12cm}\textsubscript{G}}}. Moreover, \textit{\textoverline{R\textsuperscript{+}\hspace{-0.12cm}\textsubscript{G}}} is computationally simpler and smaller than \textit{R\textsuperscript{+}\hspace{-0.12cm}\textsubscript{G}} as shown in TABLE~\ref{tab:Comparison RstarG and GRbar}. Although both computing \textit{\textoverline{R\textsuperscript{+}\hspace{-0.12cm}\textsubscript{G}}} on \textit{\textoverline{G\textsubscript{R}}} and computing \textit{R\textsuperscript{+}\hspace{-0.12cm}\textsubscript{G}} on \textit{G} find pairs of vertices by traversing a graph, the former is simpler than the latter because of the following differences:

\begin{itemize}
    \item The size of the target graph to be traversed is smaller ($\mid$\textit{\textoverline{V\textsubscript{R}}}$\mid$ $<<$ $\mid$\textit{V\textsubscript{R}}$\mid$ in general). \textit{\textoverline{G\textsubscript{R}}} is a two-level reduced graph of \textit{G}, which is reduced in size. Thus, the former target graph \textit{\textoverline{G\textsubscript{R}}} is generally smaller than the latter target graph \textit{G}.
    \item The operations performed when traversing the graph are simpler. The former traverses the graph using only operations that identify reachability. The latter, on the other hand, performs additional pattern matching operations for the labels of the edges as well. Therefore, the former operations are simpler.
\end{itemize}

TABLE~\ref{tab:Comparison RstarG and GRbar} summarizes the computational and space complexity of \textit{\textoverline{R\textsuperscript{+}\hspace{-0.12cm}\textsubscript{G}}} and \textit{R\textsuperscript{+}\hspace{-0.12cm}\textsubscript{G}}. \textit{\textoverline{R\textsuperscript{+}\hspace{-0.12cm}\textsubscript{G}}} is computed on \textit{\textoverline{G\textsubscript{R}}} after two-level graph reduction: \textit{G} $\rightarrow$ \textit{G\textsubscript{R}} $\rightarrow$ \textit{\textoverline{G\textsubscript{R}}}. \textit{R\textsuperscript{+}\hspace{-0.12cm}\textsubscript{G}} is computed using \textit{R\textsubscript{G}} after evaluating \textit{R} on \textit{G}. The main operation required when reducing the graph \textit{G} $\rightarrow$ \textit{G\textsubscript{R}} is to evaluate \textit{R} on \textit{G}. As explained in Section~\ref{subsec:Regular Path Query(RPQ)}, evaluating \textit{R} on \textit{G} is computationally simpler than and relatively negligible with evaluating \textit{R\textsuperscript{+}} on \textit{G}. Therefore, we exclude the computation for reducing the graph \textit{G} $\rightarrow$ \textit{G\textsubscript{R}} from the computational complexities of both \textit{\textoverline{R\textsuperscript{+}\hspace{-0.12cm}\textsubscript{G}}} and \textit{R\textsuperscript{+}\hspace{-0.12cm}\textsubscript{G}}. The computational complexity of evaluating \textit{R\textsuperscript{+}\hspace{-0.12cm}\textsubscript{G}} is O($\mid$\textit{V\textsubscript{R}}$\mid$$\times$$\mid$\textit{E\textsubscript{R}}$\mid$). The main operation required when reducing the graph \textit{G\textsubscript{R}} $\rightarrow$ \textit{\textoverline{G\textsubscript{R}}} is to find all SCCs of \textit{G\textsubscript{R}}. The most efficient method for this operation is known as the Tarjan's algorithm\cite{Tar72}, whose computational complexity is O($\mid$\textit{V\textsubscript{R}}$\mid$+$\mid$\textit{E\textsubscript{R}}$\mid$) $<<$ O($\mid$\textit{V\textsubscript{R}}$\mid$$\times$$\mid$\textit{E\textsubscript{R}}$\mid$). Since the overhead of reducing the graph \textit{G\textsubscript{R}} $\rightarrow$ \textit{\textoverline{G\textsubscript{R}}} is negligible compared with the computational complexity of evaluating \textit{R\textsuperscript{+}} on \textit{G\textsubscript{R}}, we exclude it from the comparison in TABLE~\ref{tab:Comparison RstarG and GRbar}. The computational complexity of computing \textit{\textoverline{R\textsuperscript{+}\hspace{-0.12cm}\textsubscript{G}}} on \textit{\textoverline{G\textsubscript{R}}} is O($\mid$\textit{\textoverline{V\textsubscript{R}}}$\mid$$\times$$\mid$\textit{\textoverline{E\textsubscript{R}}}$\mid$) which is generally smaller than O($\mid$\textit{V\textsubscript{R}}$\mid$$\times$$\mid$\textit{E\textsubscript{R}}$\mid$). Therefore, \textit{\textoverline{R\textsuperscript{+}\hspace{-0.12cm}\textsubscript{G}}} has smaller computational complexity than \textit{R\textsuperscript{+}\hspace{-0.12cm}\textsubscript{G}}. This observation is demonstrated in performance evaluation of Section~\ref{sec:Performance Evaluation}. In the worst case, \textit{\textoverline{R\textsuperscript{+}\hspace{-0.12cm}\textsubscript{G}}} and \textit{R\textsuperscript{+}\hspace{-0.12cm}\textsubscript{G}} are all vertex pairs of the target graph, so that the space complexity is O($\mid$\textit{\textoverline{V\textsubscript{R}}}$\mid$\textsuperscript{\textit{2}}) and O($\mid$\textit{V\textsubscript{R}}$\mid$\textsuperscript{\textit{2}}), respectively. By the vertex-level reduction O($\mid$\textit{\textoverline{V\textsubscript{R}}}$\mid$) $<<$ O($\mid$\textit{V\textsubscript{R}}$\mid$) in general. That is, \textit{\textoverline{R\textsuperscript{+}\hspace{-0.12cm}\textsubscript{G}}} has a generally smaller space complexity than \textit{R\textsuperscript{+}\hspace{-0.12cm}\textsubscript{G}}.

\begin{table}[b]
    \vspace{0.3cm}
    \renewcommand{\arraystretch}{1.3}
    \renewcommand{\tabcolsep}{5mm}
    \caption{Complexity of \textit{R\textsuperscript{+}\hspace{-0.12cm}\textsubscript{G}} and \textit{\textoverline{R\textsuperscript{+}\hspace{-0.12cm}\textsubscript{G}}}.}
    \label{tab:Comparison RstarG and GRbar}
    \centering{}
    \begin{tabular}{|c|c|c|}
        \hline
        Complexity & \textit{R\textsuperscript{+}\hspace{-0.12cm}\textsubscript{G}} & \textit{\textoverline{R\textsuperscript{+}\hspace{-0.12cm}\textsubscript{G}}}\tabularnewline
        \hline
        \hline
        Computational& O($\mid$\textit{V\textsubscript{R}}$\mid$$\times$$\mid$\textit{E\textsubscript{R}}$\mid$) & O($\mid$\textit{\textoverline{V\textsubscript{R}}}$\mid$$\times$$\mid$\textit{\textoverline{E\textsubscript{R}}}$\mid$)\tabularnewline
        \hline
        Space& O($\mid$\textit{V\textsubscript{R}}$\mid$\textsuperscript{\textit{2}}) & O($\mid$\textit{\textoverline{V\textsubscript{R}}}$\mid$\textsuperscript{\textit{2}})\tabularnewline
        \hline
    \end{tabular}
\end{table}

\section{Reduced Transitive Closure Sharing (RTCSharing)}
\label{sec:Multiple Regular Path Queries(MRPQs) Evaluation Method}
In this section, we propose an RPQ evaluation algorithm, which we call \textit{RTCSharing}. \textit{RTCSharing} recursively evaluates RPQs calling EvalBatchUnit, which evaluates batch units of the RPQs. \textit{RTCSharing} also share the RTC among the batch units. Section~\ref{subsec:RTCSharing} proposes \textit{RTCSharing} and Section~\ref{subsec:Evaluation of Building Blocks} optimizes the evaluation of batch units.

\subsection{RTCSharing: an RPQ Evaluation Algorithm}
\label{subsec:RTCSharing}
In \textit{RTCSharing}, we first convert the given RPQ to a logically equivalent disjunctive normal form (DNF). Since all logical formulas can be converted to a logically equivalent DNF~\cite{Dav90}, we can convert all RPQs to a logically equivalent DNF treating each outermost Kleene closure as a literal. Then, we evaluate each clause in the DNF treating it as a batch unit. The batch unit is in the form of \textit{Prefix$\cdot$R\textsuperscript{+}$\cdot$Postfix} or \textit{Prefix$\cdot$R\textsuperscript{*}$\cdot$Postfix} where \textit{Prefix} and \textit{R} are any regular expressions that can include multiple or nested Kleene closures and \textit{Postfix} is a regular expression without a Kleene closure, i.e., the \textit{R\textsuperscript{+}} or \textit{R\textsuperscript{*}} is the rightmost Kleene closure in the clause. For concise notation, we denote \textit{Prefix} and \textit{Postfix} as \textit{Pre} and \textit{Post}, respectively. If the batch unit includes multiple or nested Kleene closures, we recursively evaluate it using the result of the previous recursive step until the batch unit does not include any Kleene closures. It means that an escape hatch of the recursion is the batch unit that does not include any Kleene closures. When evaluating batch units including a Kleene closure as a common sub-query, we share the RTC among batch units using Theorem~\ref{theorem:R+G equals Cartesian product of RTC of GRbar}. That is, since we evaluate it only once, the performance degradation from the DNF transformation does not occur. We can further improve the performance by optimizing the evaluation order of the batch units, and we leave the optimization issue as a future work. We explain the details of the evaluation of the batch unit in Section~\ref{subsec:Evaluation of Building Blocks}.

\SetInd{0.4em}{0.6em}
\SetFuncSty{textup}
\SetCommentSty{textup}
\begin{algorithm}[b]
\caption{RTCSharing.}
\small
\label{alg:RTCSharing}
\LinesNumbered
\SetKwFunction{FRTCSharing}{RTCSharing}
\SetKwFunction{FConvertDNF}{ConvertRPQtoDNF}
\SetKwFunction{FEvalRPQNonKleene}{EvalRPQwithoutKC}
\SetKwFunction{FDecomposeClause}{DecomposeCL}
\SetKwFunction{FComputeIR}{Compute\_RTC}
\SetKwFunction{FEvalBBU}{EvalBatchUnit}
{
	\KwIn{Query \textit{Q}}
	\KwOut{The set of results \textit{Q\textsubscript{G}}}
	\textit{Q\textsubscript{G}} $\gets \emptyset$ \\
    
    \textit{Q\_DNF} $\gets$ \FConvertDNF{\textit{Q}}
    
    \ForEach {\textit{clause} $\mathrm{\in}$ \textit{\textit{Q\_DNF}}}
	{
	    \textit{Pre}, \textit{R}, \textit{Type}, \textit{Post} $\gets$ \FDecomposeClause{\textit{clause}}\\
	    \If{\textit{Type} == $\mathrm{NULL}$}
	    {   
	    \tcc{\textit{clause} has no Kleene closure}
	        \textit{clause\textsubscript{G}} $\gets$ \FEvalRPQNonKleene{\textit{Post}}
        }    
	    \Else 
	    {
	        \tcc{\textit{clause} has a Kleene closure}
	        \textit{Pre\textsubscript{G}} $\gets$ \FRTCSharing{\textit{Pre}}\\
	        \If{$\mathrm{the\ RTC\ for}$ \textit{R} $\mathrm{does\ not\ exist}$}
            {
                \tcc{Compute RTC: \textit{\textoverline{R\textsuperscript{+}\hspace{-0.12cm}\textsubscript{G}}} and \textit{SCC}}
                \textit{R\textsubscript{G}} $\gets$ \FRTCSharing{\textit{R}}  \\
                \textit{\textoverline{R\textsuperscript{+}\hspace{-0.12cm}\textsubscript{G}}}, \textit{SCC} $\gets$ \FComputeIR{\textit{R\textsubscript{G}}} \\
            }
	        \textit{clause\textsubscript{G}} $\gets$ \FEvalBBU{{\textit{Pre\textsubscript{G}}}, \textit{\textoverline{R\textsuperscript{+}\hspace{-0.12cm}\textsubscript{G}}}, \textit{SCC}, \textit{Type}, \textit{Post}} \\
	   }
        \textit{Q\textsubscript{G}} $\gets$ \textit{Q\textsubscript{G}} $\cup$ \textit{clause\textsubscript{G}}\\
    }
}
\end{algorithm}

Algorithm~\ref{alg:RTCSharing} shows the details of \textit{RTCSharing}. First, in line 2, we convert the given RPQ \textit{Q} to a logically equivalent DNF \textit{Q\_DNF}. Here, we treat each outermost Kleene closure as a literal. Then, we evaluate each clause in \textit{Q\_DNF} treating it as the batch unit in lines 4 to 12 and union the result in line 13. In line 4, we decompose the batch unit \textit{clause} into \textit{Pre}, \textit{R}, \textit{Type} (+, *, or NULL), and \textit{Post}. If \textit{clause} does not include a Kleene closure, both \textit{Pre} and \textit{R} are $\epsilon$, \textit{Type} is NULL, and \textit{Post} is the entire clause. In line 6, we evaluate the entire clause \textit{Post} calling EvalRPQwithoutKC, which uses any existing method~\cite{Yak16}. Otherwise, \textit{R\textsuperscript{+}} or \textit{R\textsuperscript{*}} is the rightmost Kleene closure, \textit{Type} is + or *, and \textit{Post} does not include any Kleene closures. In line 8, we evaluate \textit{Pre} recursively calling RTCSharing with \textit{Pre} as a query. If the RTC for \textit{R} exists, we reuse them. Otherwise, we compute and store them to share in lines 10 and 11. In line 10, we evaluate \textit{R} recursively calling RTCSharing with \textit{R} as a query. Since \textit{R\textsubscript{G}} is the same with the edge set of \textit{G\textsubscript{R}}, we can compute the RTC using \textit{R\textsubscript{G}}. In line 11, we compute the RTC \textit{\textoverline{R\textsuperscript{+}\hspace{-0.12cm}\textsubscript{G}}} calling Compute\_RTC that uses Tarjan's algorithm~\cite{Tar72}. In line 12, we evaluate the batch unit \textit{Pre}$\cdot$\textit{R\textsuperscript{+}}$\cdot$\textit{Post} calling EvalBatchUnit. We note that, unlike \textit{Pre}, \textit{Post} is not pre-evaluated. \textit{Post} is directly handled by EvalBatchUnit. EvalBatchUnit will be given by Algorithm~\ref{alg:RPQ Evaluation Method}. 

\begin{figure}[b]
\centering
\includegraphics[scale=0.88]{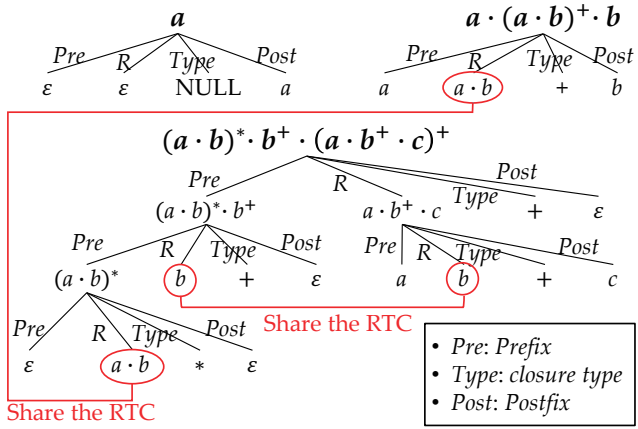}
\caption{Recursion trees of the example queries.}
\label{fig:Recursion trees of the queries in Example 7}
\end{figure}

\begin{example}
We show three examples of RPQ evaluation using \textit{RTCSharing}. To focus on the method we propose, we only consider the sequential evaluation of the given RPQs in the DNF. Fig.~\ref{fig:Recursion trees of the queries in Example 7} shows recursion trees of the RPQs.
\begin{itemize}  
    \item RPQ = \textit{a}. \\(The query does not include a Kleene closure.)
    
        We first decompose the query so that both \textit{Pre} and \textit{R} is $\epsilon$, \textit{Type} is NULL, and \textit{Post} is \textit{a}. We then evaluate the query calling EvalRPQwithoutKC in line 6.

    \item RPQ = \textit{a}$\cdot$(\textit{a}$\cdot$\textit{b})\textsuperscript{+}$\cdot$\textit{b}. \\(The query includes a Kleene closure.)
    
        We first decompose the query so that \textit{Pre} is \textit{a}, \textit{R} is \textit{a}$\cdot$\textit{b}, \textit{Type} is +, and \textit{Post} is \textit{b}. We evaluate \textit{a} in line 8 and compute the RTC for \textit{a}$\cdot$\textit{b} in lines 10 and 11. Then, we evaluate the query calling EvalBatchUnit with \textit{a\textsubscript{G}} and the RTC for \textit{a}$\cdot$\textit{b} in line 12.

    \item RPQ = (\textit{a}$\cdot$\textit{b})\textsuperscript{*}$\cdot$\textit{b}\textsuperscript{+}$\cdot$(\textit{a}$\cdot$\textit{b\textsuperscript{+}}$\cdot$\textit{c})\textsuperscript{+}. \\(The query includes multiple and nested Kleene closures.)
    
        We first decompose the query so that \textit{Pre} is (\textit{a}$\cdot$\textit{b})\textsuperscript{*}$\cdot$\textit{b}\textsuperscript{+}, \textit{R} is \textit{a}$\cdot$\textit{b\textsuperscript{+}}$\cdot$\textit{c}, \textit{Type} is +, and \textit{Post} is $\epsilon$. We evaluate (\textit{a}$\cdot$\textit{b})\textsuperscript{*}$\cdot$\textit{b}\textsuperscript{+} calling \textit{RTCSharing} recursively in line 8. In this recursive step, \textit{Pre} is (\textit{a}$\cdot$\textit{b})\textsuperscript{*} and \textit{R} is \textit{b}, \textit{Type} is + and \textit{Post} is $\epsilon$. When evaluating (\textit{a}$\cdot$\textit{b})\textsuperscript{*}, we reuse the RTC for \textit{a}$\cdot$\textit{b}, which was already computed when evaluating the above query \textit{a}$\cdot$(\textit{a}$\cdot$\textit{b})\textsuperscript{+}$\cdot$\textit{b}. When evaluating (\textit{a}$\cdot$\textit{b\textsuperscript{+}}$\cdot$\textit{c})\textsuperscript{+}, we also reuse the RTC for \textit{b}, which was already computed when evaluating (\textit{a}$\cdot$\textit{b})\textsuperscript{*}$\cdot$\textit{b}\textsuperscript{+}. Finally, we evaluate the query calling EvalBatchUnit with ((\textit{a}$\cdot$\textit{b})\textsuperscript{*}$\cdot$\textit{b}\textsuperscript{+})\textsubscript{G} and the RTC for \textit{a}$\cdot$\textit{b\textsuperscript{+}}$\cdot$\textit{c} in line 12.
\end{itemize}
\vspace{-0.6cm}
\end{example}

\subsection{Evaluation of the Batch Unit}
\label{subsec:Evaluation of Building Blocks}
When evaluating the batch unit, especially \textit{Pre}$\cdot$\textit{R\textsuperscript{+}}, \textit{useless} and \textit{redundant operations} can occur because of \textit{Pre\textsubscript{G}}. These are sources of performance degradation. Among vertex pairs in \textit{\textoverline{R\textsuperscript{+}\hspace{-0.12cm}\textsubscript{G}}}, we need only those that are connected from a vertex pair in \textit{Pre\textsubscript{G}}. That is, operations that evaluate \textit{R\textsuperscript{+}} starting from vertices that are not connected from any vertex pair in \textit{Pre\textsubscript{G}} are useless. We call them \textit{useless-1 operations}.

Vertex pairs in \textit{Pre\textsubscript{G}} whose start vertices are the same can cause duplicate results in (\textit{Pre}$\cdot$\textit{R\textsuperscript{+}})\textit{\textsubscript{G}}. The concatenations of those vertex pairs and vertex pairs in \textit{R\textsuperscript{+}\hspace{-0.12cm}\textsubscript{G}} whose end vertices are the same produce the same result. Therefore, when evaluating \textit{R\textsuperscript{+}} starting from vertex pairs in \textit{Pre\textsubscript{G}} whose start vertices are the same, operations that find vertex pairs in \textit{R\textsuperscript{+}\hspace{-0.12cm}\textsubscript{G}} whose end vertices are the same are \textit{redundant operations}. We classify the \textit{redundant operations} as \textit{redundant-1 operations} as in Definition~\ref{def:redundant 1 operations} and \textit{redundant-2 operations} as in Definition~\ref{def:redundant 2 operations}. 

\begin{definition}
\label{def:redundant 1 operations}
\textit{Redundant-1 operations} are those that find vertex pairs in \textit{R\textsuperscript{+}\hspace{-0.12cm}\textsubscript{G}} whose end vertices are the same when evaluating \textit{R\textsuperscript{+}} starting from vertex pairs in \textit{Pre\textsubscript{G}} whose start vertices are the same and end vertices belong to the same SCC.
\end{definition}

\begin{definition}
\label{def:redundant 2 operations}
\textit{Redundant-2 operations} are those that find vertex pairs in \textit{R\textsuperscript{+}\hspace{-0.12cm}\textsubscript{G}} whose end vertices are the same when evaluating \textit{R\textsuperscript{+}} starting from vertex pairs in \textit{Pre\textsubscript{G}} whose start vertices are the same but end vertices belong to different SCCs.
\end{definition}

\begin{figure}[b]
\centering
    \includegraphics[scale=0.85]{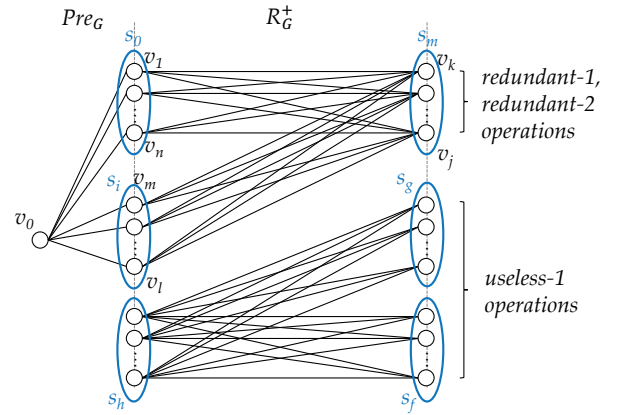}
\caption{An example of evaluating \textit{Pre}$\cdot$\textit{R\textsuperscript{+}} involving \textit{redundant} and \textit{useless operations}.}
\label{fig:an example of redundant operations and useless operations}
\end{figure} 

\begin{example}
Fig.~\ref{fig:an example of redundant operations and useless operations} shows an example evaluation of \textit{Pre}$\cdot$\textit{R\textsuperscript{+}} involving \textit{useless-1}, \textit{redundant-1}, and \textit{redundant-2 operations}. The paths satisfying \textit{R\textsuperscript{+}} from the vertices belonging to \textit{s\textsubscript{h}} to the vertices belonging to \textit{s\textsubscript{g}} and \textit{s\textsubscript{f}} are not connected from vertex pairs in \textit{Pre\textsubscript{G}}. Therefore, operations that find them are \textit{useless-1 operations}. \\
Operations that evaluate \textit{Pre}$\cdot$\textit{R\textsuperscript{+}} starting from \textit{v\textsubscript{0}} via vertices in \textit{s\textsubscript{0}} or those in  \textit{s\textsubscript{i}} produce the same results (i.e., (\textit{v\textsubscript{0}}, \textit{v\textsubscript{k}}), $\cdots$, (\textit{v\textsubscript{0}}, \textit{v\textsubscript{j}})). That is, these operations are \textit{redundant operations}. Among them, operations via vertices in the same SCC are \textit{redundant-1 operations} (e.g., the operations from \textit{v\textsubscript{0}} via vertices in \textit{s\textsubscript{0}}). Otherwise, the operations via vertices in different SCCs are \textit{redundant-2 operations} (e.g., the operations from \textit{v\textsubscript{0}} via a vertex in \textit{s\textsubscript{0}} and one in \textit{s\textsubscript{i}}).
\end{example}

These \textit{useless} and \textit{redundant operations} are the major cause of performance degradation, and we need to optimize the evaluation of batch units. To efficiently evaluate the batch units, we first formally represent the result as a relational algebra expression including \textit{\textoverline{R\textsuperscript{+}\hspace{-0.12cm}\textsubscript{G}}}. We first define the notation:
\begin{itemize}[leftmargin=0.6cm]
    \item \textit{R\textsubscript{G}}(\textit{START\_V}, \textit{END\_V}) = \{(\textit{v\textsubscript{i}}, \textit{v\textsubscript{j}})$\mid$(\textit{v\textsubscript{i}}, \textit{v\textsubscript{j}}) $\in$ \textit{R\textsubscript{G}}\}: a relation corresponding to \textit{R\textsubscript{G}} for any regular expression \textit{R}
    \item \textit{SCC}(\textit{V}, \textit{S}) = \{(\textit{v\textsubscript{i}}, \textit{s\textsubscript{j}})$\mid$\textit{v\textsubscript{i}} $\in$ \textit{s\textsubscript{j}}\}: a relation that represents the relationship between each vertex of \textit{G\textsubscript{R}} and the SCC containing the vertex
    \item \textit{\textoverline{R\textsuperscript{+}\hspace{-0.12cm}\textsubscript{G}}}(\textit{START\_S}, \textit{END\_S}) = \{(\textit{s\textsubscript{i}}, \textit{s\textsubscript{j}})$\mid$(\textit{\textoverline{v\textsubscript{i}}}, \textit{\textoverline{v\textsubscript{j}}}) $\in$ \textit{\textoverline{R\textsuperscript{+}\hspace{-0.12cm}\textsubscript{G}}}\}: a relation corresponding to \textit{\textoverline{R\textsuperscript{+}\hspace{-0.12cm}\textsubscript{G}}}, the transitive closure of \textit{\textoverline{G\textsubscript{R}}}
\end{itemize}

Using the notation above, we represent an equivalent (1) as Lemma~\ref{lemma:ST eq}.

\begin{lemma}
\label{lemma:ST eq}
\begin{equation}
\label{eq:Relational Algebra for ST}
    \begin{aligned}
    (\textit{A}\cdot \textit{B})\textsubscript{\textit{G}}&(\textit{START\_V}, \textit{END\_V})\\
    = &\pi_{\textit{A}\textsubscript{\textit{G}}.\textit{START\_V}, \textit{B}\textsubscript{\textit{G}}.\textit{END\_V}} 
    \left[\vphantom{\bowtie_{R^+_G}}
        \textit{A}\textsubscript{\textit{G}}(\textit{START\_V}, \textit{END\_V})
    \right.
    \\
    &\left.
        \bowtie_{\textit{A}\textsubscript{\textit{G}}.\textit{END\_V} = \textit{B}\textsubscript{\textit{G}}.\textit{START\_V}} \textit{B}\textsubscript{\textit{G}}(\textit{START\_V}, \textit{END\_V}) 
     \vphantom{\bowtie_{R^+_G}}\right]
    \end{aligned}
\end{equation}
\end{lemma}

\begin{proof}
Because of space limitations in proof, we omit the attribute names of the relations.
\begin{itemize}
    \item \textit{(A$\cdot$B)\textsubscript{G}} $\subseteq$ $\pi$\textsubscript{\textit{A\textsubscript{G}}.\textit{START\_V}, \textit{B\textsubscript{G}}.\textit{END\_V}}(\textit{A\textsubscript{G}}$\bowtie$\textsubscript{\textit{A\textsubscript{G}.END\_V=B\textsubscript{G}.START\_V}}\textit{B\textsubscript{G}}).
    
    Suppose \textit{(A$\cdot$B)\textsubscript{G}} $\ni$ (\textit{v\textsubscript{i}}, \textit{v\textsubscript{j}}). Then, there exists a path \textit{p}(\textit{v\textsubscript{i}}, \textit{v\textsubscript{j}}) satisfying \textit{A$\cdot$B} on \textit{G}. \textit{p}(\textit{v\textsubscript{i}}, \textit{v\textsubscript{j}}) can be decomposed into a path \textit{p}(\textit{v\textsubscript{i}}, \textit{v\textsubscript{k}}) satisfying \textit{A} and a path \textit{p}(\textit{v\textsubscript{k}}, \textit{v\textsubscript{j}}) satisfying \textit{B}. Because (\textit{v\textsubscript{i}}, \textit{v\textsubscript{k}}) $\in$ \textit{A\textsubscript{G}} and (\textit{v\textsubscript{k}}, \textit{v\textsubscript{j}}) $\in$ \textit{B\textsubscript{G}}, (\textit{v\textsubscript{i}}, \textit{v\textsubscript{j}}) $\in$ $\pi$\textsubscript{\textit{A\textsubscript{G}}.\textit{START\_V}, \textit{B\textsubscript{G}}.\textit{END\_V}}(\textit{A\textsubscript{G}}$\bowtie$\textsubscript{\textit{A\textsubscript{G}.END\_V=B\textsubscript{G}.START\_V}}\textit{B\textsubscript{G}}).

    \item \textit{(A$\cdot$B)\textsubscript{G}} $\supseteq$ $\pi$\textsubscript{\textit{A\textsubscript{G}}.\textit{START\_V}, \textit{B\textsubscript{G}}.\textit{END\_V}}(\textit{A\textsubscript{G}}$\bowtie$\textsubscript{\textit{A\textsubscript{G}.END\_V=B\textsubscript{G}.START\_V}}\textit{B\textsubscript{G}}).

    Suppose $\pi$\textsubscript{\textit{A\textsubscript{G}}.\textit{START\_V}, \textit{B\textsubscript{G}}.\textit{END\_V}}(\textit{A\textsubscript{G}}$\bowtie$\textsubscript{\textit{A\textsubscript{G}.END\_V=B\textsubscript{G}.START\_V}}\textit{B\textsubscript{G}}) $\ni$ (\textit{v\textsubscript{i}}, \textit{v\textsubscript{j}}). Then, there must exist (\textit{v\textsubscript{i}}, \textit{v\textsubscript{k}}) and (\textit{v\textsubscript{k}}, \textit{v\textsubscript{j}}) where (\textit{v\textsubscript{i}}, \textit{v\textsubscript{k}}) $\in$ \textit{A\textsubscript{G}}, (\textit{v\textsubscript{k}}, \textit{v\textsubscript{j}}) $\in$ \textit{B\textsubscript{G}}. Because there exists a path \textit{p}(\textit{v\textsubscript{i}}, \textit{v\textsubscript{k}}) satisfying \textit{A} and a path \textit{p}(\textit{v\textsubscript{k}}, \textit{v\textsubscript{j}}) satisfying \textit{B} on \textit{G}, there must exist a \textit{p}(\textit{v\textsubscript{i}}, \textit{v\textsubscript{j}}) satisfying \textit{A$\cdot$B}. Therefore, (\textit{v\textsubscript{i}}, \textit{v\textsubscript{j}}) $\in$\textit{(A$\cdot$B)\textsubscript{G}} holds.
\end{itemize}
\vspace{-0.4cm}
\end{proof}

Theorem~\ref{theorem:R+G relational algebra exprssion} shows a relational algebra expression for evaluation of RPQ \textit{R\textsuperscript{+}} on \textit{G} using the RTC. The result of RPQ \textit{R\textsuperscript{+}} on \textit{G} is represented as a relational algebra expression in Theorem~\ref{theorem:R+G relational algebra exprssion}. Equation~(\ref{eq:Relational Algebra for Kleene Closure}) represents the union of Cartesian products of \textit{s\textsubscript{i}} and \textit{s\textsubscript{j}} for every vertex pair(\textit{\textoverline{v\textsubscript{i}}}, \textit{\textoverline{v\textsubscript{j}}})(i.e., an element in the RTC). That is, we can efficiently evaluate \textit{R\textsuperscript{+}} by using the RTC.

\begin{theorem}
\label{theorem:R+G relational algebra exprssion}
The result of RPQ \textit{R\textsuperscript{+}} on \textit{G} is represented as a relational algebra expression in (\ref{eq:Relational Algebra for Kleene Closure}) for any RPQ \textit{R}. Here, $\rho$\textsubscript{\textit{SSCC}} and $\rho$\textsubscript{\textit{ESCC}} are renaming operations. 

\begin{equation}
\label{eq:Relational Algebra for Kleene Closure}
    \begin{aligned}
    &\textit{R\textsuperscript{+}\hspace{-0.12cm}\textsubscript{G}}(\textit{START\_V}, \textit{END\_V})\\
    = &\pi_{\textit{SSCC.V}, \textit{ESCC.V}}
    \left[\vphantom{\bowtie_{R^+_G}}
        \rho_{\textit{SSCC}}(\textit{SCC}(\textit{V}, \textit{S}))
    \right.\\
    &\bowtie_{\textit{S} = \textit{START\_S}} \textit{\textoverline{R\textsuperscript{+}\hspace{-0.12cm}\textsubscript{G}}}(\textit{START\_S}, \textit{END\_S}) \\
    &\left.
        \bowtie_{\textit{END\_S} = \textit{S}} \rho_{\textit{ESCC}}(\textit{SCC}(\textit{V}, \textit{S}))
    \vphantom{\bowtie_{R^+_G}}\right]
    \end{aligned}
\end{equation}
\end{theorem}

\begin{proof}
Theorem~\ref{theorem:R+G equals Cartesian product of RTC of GRbar} can be represented as the following tuple relational calculus expression. 
\begin{align}
\nonumber
\textit{R\textsuperscript{+}\hspace{-0.12cm}\textsubscript{G}} &= \{\textit{res} \mid (\exists \textit{rtc})(\exists \textit{sscc})(\exists \textit{escc})\ \textit{\textoverline{R\textsuperscript{+}\hspace{-0.12cm}\textsubscript{G}}}(rtc) \ \mathrm{AND} \\ 
\nonumber
&\ \ \ \ \ \ \ \ \ \ \ \ \ \ \ \textit{SCC}(\textit{sscc})\ \mathrm{AND} \ \textit{SCC}(\textit{escc}) \ \mathrm{AND} \\
&\textit{sscc}[\textit{S}] = \textit{rtc}[\textit{START\_S}] \ \mathrm{AND} \ \textit{rtc}[\textit{END\_S}] = \textit{escc}[\textit{S}] \ \mathrm{AND}  \tag*{/* (\textit{v\textsubscript{i}}, \textit{v\textsubscript{j}}) $\in$ \textit{s\textsubscript{k}}$\times$\textit{s\textsubscript{l}} */}\\
&\textit{res}[\textit{START\_V}] = \textit{sscc}[\textit{V}] \ \mathrm{AND} \ \textit{res}[\textit{END\_V}] = \textit{escc}[\textit{V}]\}  \tag*{/* projection */}
\end{align}
Equation~(2) is derived from the above relational calculus expression in a straightforward manner~\cite{Ull88}.
\end{proof}

We now expand the result of the batch unit to eliminate \textit{useless} and \textit{redundant operations}. Using Lemma~4 and (\ref{eq:Relational Algebra for ST}), the result of the batch unit can be represented as in (3) to (5). Using (\ref{eq:Relational Algebra for Kleene Closure}) we can expand (4) as in (7) to (9). When evaluating (4), we can eliminate \textit{useless-1 operations} by using \textit{\textoverline{R\textsuperscript{+}\hspace{-0.12cm}\textsubscript{G}}} instead of the entire \textit{R\textsuperscript{+}\hspace{-0.12cm}\textsubscript{G}} and evaluating \textit{R\textsuperscript{+}} starting only from tuples (i.e., vertex pairs) in \textit{Pre\textsubscript{G}}(\textit{START\_V}, \textit{END\_V}).

\begin{align}
    (\textit{Pre}&\cdot \textit{R\textsuperscript{+}}\cdot \textit{Post})\textsubscript{\textit{G}}(\textit{START\_V}, \textit{END\_V})\nonumber\\
    =&\pi_{\textit{Pre\textsubscript{G}.START\_V}, \textit{Post\textsubscript{G}.END\_V}}
    \left[\vphantom{\bowtie_{R^+_G}}
        \textit{Pre\textsubscript{G}}(\textit{START\_V}, \textit{END\_V})
    \right.\\
    &\bowtie_{\textit{Pre\textsubscript{G}.END\_V}=\textit{R\textsuperscript{+}\hspace{-0.12cm}\textsubscript{G}.START\_V}} \textit{R\textsuperscript{+}\hspace{-0.12cm}\textsubscript{G}}(\textit{START\_V}, \textit{END\_V})\\
    &\left.
        \bowtie_{\textit{R\textsuperscript{+}\hspace{-0.12cm}\textsubscript{G}.END\_V}=\textit{Post\textsubscript{G}.START\_V}} \textit{Post\textsubscript{G}}(\textit{START\_V}, \textit{END\_V})
    \vphantom{\bowtie_{R^+_G}}\right]
\end{align}
\begin{align}
    (\textit{Pre}&\cdot \textit{R\textsuperscript{+}}\cdot \textit{Post})\textsubscript{\textit{G}}(\textit{START\_V}, \textit{END\_V})\nonumber\\
    =&\pi_{\textit{Pre\textsubscript{G}.START\_V}, \textit{Post\textsubscript{G}.END\_V}}
    \left[\vphantom{\bowtie_{R^+_G}}
        \textit{Pre\textsubscript{G}}(\textit{START\_V}, \textit{END\_V})
    \right.\\
    &\bowtie_{\textit{END\_V=V}} \textit{SCC}(\textit{V}, \textit{S}) \\
    &\bowtie_{\textit{S=START\_S}} \textit{\textoverline{R\textsuperscript{+}\hspace{-0.12cm}\textsubscript{G}}}(\textit{START\_S}, \textit{END\_S}) \\
    &\bowtie_{\textit{END\_S=S}} \textit{SCC}(\textit{V}, \textit{S})\\
    &\left.
        \bowtie_{\textit{V=START\_V}} \textit{Post\textsubscript{G}}(\textit{START\_V}, \textit{END\_V})
    \vphantom{\bowtie_{R^+_G}}\right]
\end{align}

To eliminate \textit{redundant-1} and \textit{redundant-2 operations}, we union the intermediate results, i.e., the RTCs, at each join step. For each vertex pair (\textit{v\textsubscript{i}}, \textit{v\textsubscript{j}}) in \textit{Pre\textsubscript{G}}, (7) finds the SCC \textit{s\textsubscript{k}} containing \textit{v\textsubscript{j}} and returns (\textit{v\textsubscript{i}}, \textit{s\textsubscript{k}}) as the result. The intermediate results of (7) are duplicated for vertex pairs whose start vertices are the same and end vertices belong to the same SCC. Thus, we can eliminate \textit{redundant-1 operations} in the next join step by unioning (i.e., eliminating duplicates in) the intermediate results of (7). For each pair of vertex and SCC (\textit{v\textsubscript{i}}, \textit{s\textsubscript{k}}) in \textit{Pre\textsubscript{G}} $\bowtie$ SCC, (8) finds the SCC \textit{s\textsubscript{l}} that is reachable from \textit{s\textsubscript{k}} and returns (\textit{v\textsubscript{i}}, \textit{s\textsubscript{l}}) as the result. If, in the results of (7), there exist paths from vertex pairs whose start vertices are the same but end vertices belong to different SCCs to vertices belonging to the same SCC, the intermediate results of (8) for those paths are duplicated. Therefore, we find unique SCCs containing end vertices of the paths satisfying \textit{R\textsuperscript{+}} by unioning (i.e., eliminating duplicates in) the intermediate results of (8). Thus, we can eliminate subsequent \textit{redundant-2 operations}.
Because sets of vertices belonging to different SCCs are mutually disjoint, if the results of (8) are distinct, there are no duplicates in the result of (9). That is, an operation that checks duplicates for union when performing (9) is useless. We call them \textit{useless-2 operations}. We can eliminate \textit{useless-2 operations} by not performing duplicate checks. 

\begin{example}
Fig.~\ref{fig:an example of eliminating redundant operations and useless operations} shows an example of evaluating \textit{Pre$\cdot$R\textsuperscript{+}} where \textit{useless-1}, \textit{redundant-1}, and \textit{redundant-2 operations} are eliminated. We evaluate \textit{R\textsuperscript{+}} starting from vertex pairs in \textit{Pre\textsubscript{G}} using \textit{\textoverline{R\textsuperscript{+}\hspace{-0.12cm}\textsubscript{G}}} instead of \textit{R\textsuperscript{+}\hspace{-0.12cm}\textsubscript{G}}. Therefore, the paths satisfying \textit{R\textsuperscript{+}} from the vertices belonging to \textit{s\textsubscript{n+1}} to the vertices belonging to \textit{s\textsubscript{k}} and \textit{s\textsubscript{k+1}} are not found. Thus, \textit{useless-1 operations} are eliminated.\\
We next eliminate duplicates of the intermediate results of (7). Therefore, vertex pairs in (6) whose start vertices are the same and end vertices belong to the same SCC become one pair of vertex and SCC in (7). That is, (\textit{v\textsubscript{0}}, \textit{v\textsubscript{1}}), $\cdots$, (\textit{v\textsubscript{0}}, \textit{v\textsubscript{n}}) in (6) become (\textit{v\textsubscript{0}}, \textit{s\textsubscript{0}}) in (7). Therefore, when evaluating \textit{R\textsuperscript{+}} starting from (\textit{v\textsubscript{0}}, \textit{v\textsubscript{1}}), $\cdots$, (\textit{v\textsubscript{0}}, \textit{v\textsubscript{n}}) in \textit{Pre\textsubscript{G}}, we find \textit{v\textsubscript{k}}, $\cdots$, \textit{v\textsubscript{j}} only once by going through \textit{s\textsubscript{0}}. Thus, \textit{redundant-1 operations} are eliminated. \\
We also eliminate duplicates of the intermediate results of (8). Therefore, if there exist paths from vertex pairs whose start vertices are the same but end vertices belong to different SCCs to vertices belonging to a same SCC, those vertex pairs in (6) become one pair of vertex and SCC in (8). That is, (\textit{v\textsubscript{0}}, \textit{v\textsubscript{1}}) and (\textit{v\textsubscript{0}}, \textit{v\textsubscript{m}}) in (6) become (\textit{v\textsubscript{0}}, \textit{s\textsubscript{m}}) in (8). Therefore, when evaluating \textit{Pre$\cdot$R\textsuperscript{+}} starting from \textit{v\textsubscript{0}} via each vertex belonging to \textit{s\textsubscript{0}} and \textit{s\textsubscript{n}}, respectively, we find \textit{v\textsubscript{k}}, $\cdots$, \textit{v\textsubscript{j}} only once by going through the unique \textit{s\textsubscript{m}}. Thus, subsequent \textit{redundant-2 operations} are eliminated. 
\end{example}

\begin{figure}[b] 
\centering
    \includegraphics[scale=0.88]{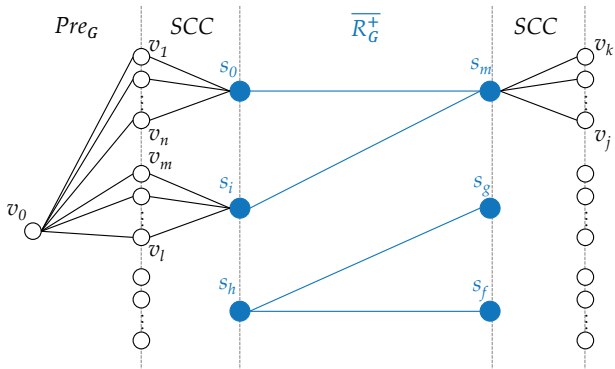}
\caption{An example of evaluating \textit{Pre$\cdot$R\textsuperscript{+}} eliminating \textit{redundant} and \textit{useless operations}.}
\label{fig:an example of eliminating redundant operations and useless operations}
\end{figure}

Algorithm~\ref{alg:RPQ Evaluation Method} shows EvalBatchUnit that evaluates a batch unit in an optimized way of eliminating \textit{useless} and \textit{redundant operations} explained so far. First, in line 1 to 3, we initialize variables storing the results of (7) to (10). Here, we deal with \textit{Pre$\cdot$R\textsuperscript{*}$\cdot$Post} by initializing \textit{ResEq9} with \textit{Pre\textsubscript{G}}. Then, in line 4, we eliminate \textit{useless-1 operations}. By line 4, lines 5 to 12 that evaluate \textit{R\textsuperscript{+}} are executed only for vertex pairs existing in \textit{Pre\textsubscript{G}}. Otherwise, these operations are \textit{useless-1 operations}. In lines 5 to 7, we eliminate \textit{redundant-1 operations}. In line 5, for each vertex pair (\textit{v\textsubscript{i}}, \textit{v\textsubscript{j}}) we find the SCC \textit{s\textsubscript{k}} of \textit{G\textsubscript{R}} containing \textit{v\textsubscript{j}}. In lines 6 and 7, we union (\textit{v\textsubscript{i}}, \textit{s\textsubscript{j}}) into \textit{ResEq7}. Then, lines 8 to 12 are executed only when (\textit{v\textsubscript{i}}, \textit{s\textsubscript{j}}) does not exist in \textit{ResEq7}. Otherwise, these operations are \textit{redundant-1 operations}. In lines 8 to 10, we eliminate \textit{redundant-2 operations}. In line 8, we find \textit{s\textsubscript{k}}, which is reachable from \textit{s\textsubscript{j}}. In lines 9 and 10, we union (\textit{v\textsubscript{i}}, \textit{s\textsubscript{k}}) into \textit{ResEq8}. Lines 11 and 12 are executed only when (\textit{v\textsubscript{i}}, \textit{s\textsubscript{k}}) does not exist in \textit{ResEq8}. Otherwise, these operations are \textit{redundant-2 operations}. In lines 11 and 12, we eliminate \textit{useless-2 operations}. Here, we add (\textit{v\textsubscript{i}}, \textit{v\textsubscript{k}}) into \textit{ResEq9} for each vertex \textit{v\textsubscript{k}} contained in \textit{s\textsubscript{k}} without duplicate checks, which are \textit{useless-2 operations}. Lines 13 to 16 that evaluate \textit{Post} are executed only for vertex pairs existing in \textit{ResEq9} (i.e., (\textit{Pre$\cdot$R\textsuperscript{+}})\textsubscript{\textit{G}}). EvalRestrictedRPQ(\textit{Post}, \textit{v\textsubscript{k}}) finds paths satisfying \textit{Post} from the vertex \textit{v\textsubscript{k}} on \textit{G} and returns (\textit{v\textsubscript{k}}, \textit{v\textsubscript{l}}) for each path \textit{p}(\textit{v\textsubscript{k}}, \textit{v\textsubscript{l}}).

\begin{algorithm}[t]
\caption{EvalBatchUnit.}
\small
{
\label{alg:RPQ Evaluation Method}
\LinesNumbered
\SetKwProg{Fn}{Function}{}{}
\SetKwFunction{FInitRes}{Init\_ResEq10}
\SetKwFunction{FEvalRestrictedRPQ}{EvalRestrictedRPQ}
\SetFuncSty{textup}
{
	\KwIn{\textit{Pre\textsubscript{G}}, \textit{\textoverline{R\textsuperscript{+}\hspace{-0.12cm}\textsubscript{G}}}, \textit{SCC}, \textit{Type}, \textit{Post}}
	\KwOut{\textit{ResEq10}}
	
    \tcc{Initialize variables storing the results of (7) to (10)}
	\textit{ResEq7}, \textit{ResEq8}, \textit{ResEq9}, \textit{ResEq10} $\gets$ $\emptyset$\\ 
	\If{\textit{Type} is *}
	{
	    \textit{ResEq9} $\gets$ \textit{Pre\textsubscript{G}} 	   \tcp*[f]{Initialization for \textit{Pre$\cdot$R\textsuperscript{*}$\cdot$Post}}
	}
	  
	\tcc{Compute \textit{Pre\textsubscript{G}} $\bowtie$ \textit{R\textsuperscript{+}\hspace{-0.12cm}\textsubscript{G}}: (7) to (9)}
    \ForEach {$\mathrm{(}$\textit{v\textsubscript{i}}, \textit{v\textsubscript{j}}$\mathrm{)}$ $\in$ \textit{Pre\textsubscript{G}}}{
	    \tcc{Eliminate \textit{useless-1 ops}.}
	    \textit{s\textsubscript{j}} $\gets$ $\pi_{\textit{S}}$($\sigma_{\textit{V = v\textsubscript{j}}}$\textit{SCC})\\
	    \If(\tcp*[f]{duplicate check for (7)}){$\mathrm{(}$\textit{v\textsubscript{i}}, \textit{s\textsubscript{j}}$\mathrm{)}$ $\not\in$ \textit{ResEq7}}
	    {

	        Insert (\textit{v\textsubscript{i}}, \textit{s\textsubscript{j}}) into \textit{ResEq7}
	        \tcp*[f]{union the result of (7)}\\
	        \tcc{Eliminate \textit{redundant-1 ops}.}
	        \ForEach{$\mathrm{(}$\textit{s\textsubscript{j}}, \textit{s\textsubscript{k}}$\mathrm{)}$ $\in$ $\sigma_{\textit{START\_S = s\textsubscript{j}}}$\textit{\textoverline{R\textsuperscript{+}\hspace{-0.12cm}\textsubscript{G}}}}{
    	        \If(\tcp*[f]{duplicate check for (8)}){$\mathrm{(}$\textit{v\textsubscript{i}}, \textit{s\textsubscript{k}}$\mathrm{)}$ $\not\in$ \textit{ResEq8}} 
	            {
	                Insert (\textit{v\textsubscript{i}}, \textit{s\textsubscript{k}}) into \textit{ResEq8}\tcp*[f]{union the result of (8)}
	                \tcc{Eliminate \textit{redundant-2 ops}.}
	                \ForEach{$\mathrm{(}$\textit{s\textsubscript{k}}, \textit{v\textsubscript{k}}$\mathrm{)}$ $\in$ $\sigma_{\textit{S = s\textsubscript{k}}}$ \textit{SCC}$\mathrm{(}$\textit{S}, \textit{V}$\mathrm{)}$}{
	                    \tcc{Eliminate \textit{useless-2 ops}.}
        	            Insert (\textit{v\textsubscript{i}}, \textit{v\textsubscript{k}}) into \textit{ResEq9}\tcp*[f]{add the result of (9)} 
	                }
	            }
	        }
	    }
	}
	\tcc{Compute (\textit{Pre}$\cdot$\textit{R\textsuperscript{+}})\textsubscript{\textit{G}} $\bowtie$ \textit{Post\textsubscript{G}}: (10)}
	\ForEach {$\mathrm{(}$\textit{v\textsubscript{i}}, \textit{v\textsubscript{k}}$\mathrm{)}$ $\in$ \textit{ResEq9}}{
        \ForEach {$\mathrm{(}$\textit{v\textsubscript{k}}, \textit{v\textsubscript{l}}$\mathrm{)}$ $\in$ \FEvalRestrictedRPQ{\textit{Post}, \textit{v\textsubscript{k}}} 
        }{
            \If(\tcp*[f]{duplicate check for (10)}){$\mathrm{(}$\textit{v\textsubscript{i}}, \textit{v\textsubscript{l}}$\mathrm{)}$ $\not\in$ \textit{ResEq10}}
            {
                Insert (\textit{v\textsubscript{i}}, \textit{v\textsubscript{l}}) into \textit{ResEq10}
                \tcp*[f]{union the result of (10)}
            }
	    }
	}
}
}
\end{algorithm}

\section{Performance Evaluation}
\label{sec:Performance Evaluation}
\normalsize
In this section, we focus on evaluating multiple RPQs in the form of the batch unit. With this, we can easily control parameters synthetically to generate complete test cases. 

\subsection{Experimental Environment}
\label{subsec:Experiment Environment}
We use synthetic datasets generated by the RMAT model~\cite{Cha04} using TrillionG~\cite{Par17} and four real datasets: Yago2s~\cite{Yag18}, Robots~\cite{Rob18}, Advogato~\cite{Adv18}, and Youtube\_Sampled~\cite{You18}. Using TrillionG we generate synthetic graphs of various average vertex degrees per label (i.e., $\frac{\mid\textit{E}\mid}{\mid \textit{V}\mid\mid\mathit{\Sigma}\mid}$) where the other characteristics except the degree are kept the same. Since TrillionG generates edge-unlabeled, directed multigraphs, we randomly added a label to each edge to make edge-labeled graphs. We denote the RMAT graph with 2\textsuperscript{13} vertices and 2\textsuperscript{\textit{N}+13} edges by RMAT\_\textit{N}. Yago2s, Robots, and Advogato are edge-labeled, directed multigraphs. Youtube\_Sampled is a subset of Youtube, which is an edge-labeled, undirected multigraph. Since the data subject to RPQs is a directed graph, we randomly added a direction to each edge of Youtube\_Sampled. We construct Youtube\_Sampled from Youtube using random vertex sampling. \textit{V} is the set of randomly sampled vertices, and \textit{E} is the set of edges between sampled vertices. We simply denote Youtube\_Sampled by Youtube. TABLE~\ref{tab:Dataset Statistics} summarizes the statistics of the datasets. 

\begin{table}[b]
    \renewcommand{\tabcolsep}{1.25mm}
    \caption{Statistics of datasets used in the experiments.}
    \label{tab:Dataset Statistics}
    \centering{}
    \begin{tabular}{|c|c|c|c|c|c|}
        \hline
        \multicolumn{2}{|c|}{Dataset} & $\mid$\textit{V}$\mid$ & $\mid$\textit{E}$\mid$ & $\mid$$\mathit{\Sigma}$$\mid$ & $\frac{\mid\textit{E}\mid}{\mid \textit{V}\mid\mid\mathit{\Sigma}\mid}$ \\
        \hline
        \hline
        \multirow{4}{*}{\minitab[c]{Real\\graph datasets}} & Yago2s & 108,048,761 & 244,796,155 & 104 & 0.02\\
        \cline{2-6}
        &Robots & 1,725 & 3,596 & 4 & 0.52\\
        \cline{2-6}
        &Advogato & 6,541 & 51,127 & 3 & 2.61\\
        \cline{2-6}
        &Youtube & 1,600 & 91,343 & 5 & 11.42\\
        \hline
        \multirow{2}{*}{\minitab[c]{Synthetic\\graph datasets}} & \multirow{2}{*}{\minitab[c]{RMAT\_\textit{N}\\(\textit{N} = 0..6)}} & \multirow{2}{*}{\minitab[c]{2\textsuperscript{\scriptsize 13}}} & \multirow{2}{*}{\minitab[c]{2\textsuperscript{\scriptsize \textit{N}+13}}} & \multirow{2}{*}{4} & \multirow{2}{*}{\minitab[c]{2\textsuperscript{\scriptsize \textit{N}-2}}}\\
        &&&&&\\
        \hline
    \end{tabular}
\end{table}

We use synthetic multiple RPQ sets where each RPQ is in the form of the batch unit. To create a controlled environment, we simulate the effects of \textit{Pre} and \textit{Post} using single labels and model \textit{R} as a concatenation of labels in $\mathit{\Sigma}$ (i.e., a clause without Kleene closure) whose length varies from 1 to 3. First, we randomly select a total of 90 \textit{R}s (one for each multiple RPQ set), 10 for each length, and then, randomly select from $\mathit{\Sigma}$ pairs of \textit{Pre} and \textit{Post} for each \textit{R}. The number of RPQs in each multiple RPQ set is 1, 2, 4, 6, 8, and 10, and a larger multiple RPQ set contains smaller multiple RPQ sets.

The multiple RPQ evaluation methods to be tested for comparison are as follows. Since the source codes for \textit{NoSharing} and \textit{FullSharing} have not been released, we implement them ourselves based on the literature~\cite{Yak16, Abu17}. 

\begin{itemize}
    \item \textit{RTCSharing}(\textit{RTC}): A method sharing \textit{\textoverline{R\textsuperscript{+}\hspace{-0.12cm}\textsubscript{G}}} among RPQs.
    \item \textit{NoSharing}(\textit{No}): A method individually evaluating RPQs using the single RPQ evaluation method proposed by Yakovets et al.~\cite{Yak16}.
    \item \textit{FullSharing}(\textit{Full}): A method sharing \textit{R\textsuperscript{+}\hspace{-0.12cm}\textsubscript{G}} among RPQs proposed by Abul-Basher~\cite{Abu17}
\end{itemize}

All experiments have been conducted on a Linux (kernel version: 2.6.32) machine equipped with an Intel Core i7-7700 CPU and 64GB main memory. All the multiple RPQ evaluation methods used in the experiment including \textit{RTCSharing} have been implemented in C++. 

\subsection{Performance Evaluation}
\label{subsec:Expermient for Multiple RPQs}
In the experiment, we compare the performance of evaluating multiple RPQs whose common sub-query is a Kleene plus \textit{R\textsuperscript{+}}. Section~\ref{subsubsec:Experiment 1} compares the performance as we vary the average vertex degrees per label (i.e., $\frac{\mid\textit{E}\mid}{\mid \textit{V}\mid\mid\mathit{\Sigma}\mid}$). In Section~\ref{subsubsec:Experiment 2}, we vary the number of RPQs constituting a multiple RPQ set to see the amortization effect of sharing the data, \textit{R\textsuperscript{+}\hspace{-0.12cm}\textsubscript{G}} or \textit{\textoverline{R\textsuperscript{+}\hspace{-0.12cm}\textsubscript{G}}}, among RPQs.

The performance metrics are multiple RPQ sets' average query response time (\textit{query response time} in short) and average size of data shared among RPQs (\textit{shared data size} in short). The query response time includes the time taken 1) to construct the two-level reduced graph in \textit{RTCSharing}, 2) to compute the shared data (\textit{\textoverline{R\textsuperscript{+}\hspace{-0.12cm}\textsubscript{G}}} in \textit{RTCSharing} or \textit{R\textsuperscript{+}\hspace{-0.12cm}\textsubscript{G}} in \textit{FullSharing}), and 3) to complete the evaluation of all RPQs. We also divide the query evaluation into three parts and compare each part between \textit {RTCSharing} and \textit {FullSharing} to show the effects of different aspects of \textit{RTCSharing} on performance. First, to show that \textit{RTCSharing} is simpler than \textit{FullSharing} for computing the data shared among RPQs, we compare the computation time of \textit{\textoverline{R\textsuperscript{+}\hspace{-0.12cm}\textsubscript{G}}} in \textit{RTCSharing} with that of \textit{R\textsuperscript{+}\hspace{-0.12cm}\textsubscript{G}} in \textit{FullSharing}. This computation time is denoted by \textit{Shared\_Data}. Since the two methods compute \textit{R\textsubscript{G}} identically, we exclude the computation time of \textit{R\textsubscript{G}} from Shared\_Data of both methods. Second, to show the effects of avoiding \textit{redundant} and \textit{useless operations} by representing each RPQ as a relational algebra expression and optimizing its evaluation, we compare the computation time of 
\textit{Pre\textsubscript{G}}(\textit{START\_V}, \textit{END\_V}) $\bowtie$ \textit{R\textsuperscript{+}\hspace{-0.12cm}\textsubscript{G}}(\textit{START\_V}, \textit{END\_V}) (denoted by \textit{Pre\textsubscript{G}} $\bowtie$ \textit{R\textsuperscript{+}\hspace{-0.12cm}\textsubscript{G}}) of each method. The only difference between two methods in \textit{Pre\textsubscript{G}} $\bowtie$ \textit{R\textsuperscript{+}\hspace{-0.12cm}\textsubscript{G}} is the optimization related to \textit{useless} and \textit{redundant operations}. Finally, we compare the computation time of the remainder evaluation (i.e., computing \textit{Pre\textsubscript{G}}, \textit{R\textsubscript{G}}, and computing (\textit{Pre$\cdot$R\textsuperscript{+}$\cdot$Post})\textsubscript{\textit{G}}
from (\textit{Pre$\cdot$R\textsuperscript{+}})\textsubscript{\textit{G}}) for which the comparison methods operate identically. Since \textit{NoSharing} does not share data among RPQs, for \textit{NoSharing}, we present only the overall query response time. The shared data size is the number of pairs in \textit{\textoverline{R\textsuperscript{+}\hspace{-0.12cm}\textsubscript{G}}} for \textit{RTCSharing} or that for \textit{R\textsuperscript{+}\hspace{-0.12cm}\textsubscript{G}} for \textit{FullSharing}. Since the size of data shared among RPQs is not dependent on the number of RPQs, we present it only in Section~\ref{subsubsec:Experiment 1}.

\subsubsection{Experiment 1: Average vertex degree per label is varied.}
\label{subsubsec:Experiment 1}
We experiment the graphs having different average vertex degree per label (\textit{vertex degree} in short). We use synthetic and real datasets described in TABLE~\ref{tab:Dataset Statistics}. We use multiple RPQ sets consisting of 4(median) RPQs among the multiple RPQ sets generated as described in Section~\ref{subsec:Experiment Environment}.



\begin{figure}[b] 
\centering
    \includegraphics[scale=1]{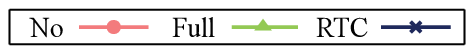}
\begin{subfigure}[t]{.49\columnwidth}
    \includegraphics[scale=1]{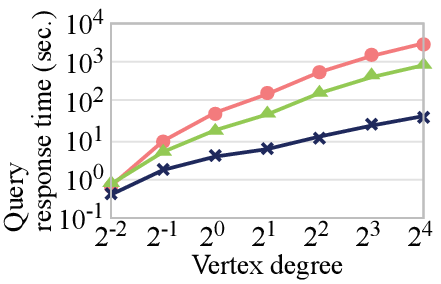}
    \caption{Synthetic datasets.}
\end{subfigure}
\begin{subfigure}[t]{.49\columnwidth}
    \includegraphics[scale=1]{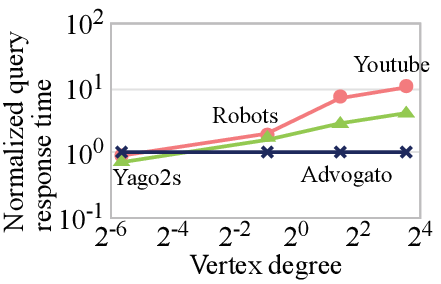}
    \caption{Real datasets.}
\end{subfigure}
\caption{Computational performances of \textit{No}, \textit{Full}, and \textit{RTC} as the vertex degree is varied ($\sharp$ RPQs = 4).}
\label{fig: Computational performances varying vertex degrees}
\end{figure} 

\begin{figure*}[t] 
\begin{subfigure}[t]{\linewidth}
\centering
\includegraphics[scale=1]{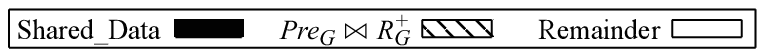}
\end{subfigure}
\begin{subfigure}[t]{1.2\columnwidth}
\includegraphics[scale=1]{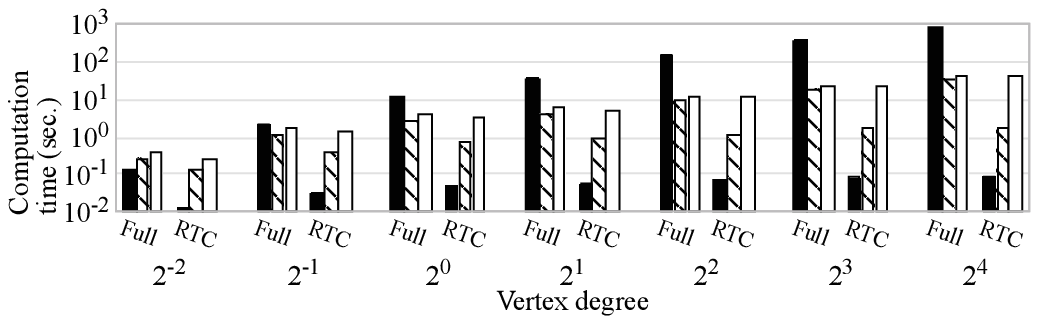}
\caption{Synthetic datasets.}
\end{subfigure}
\begin{subfigure}[t]{0.8\columnwidth}
\includegraphics[scale=1]{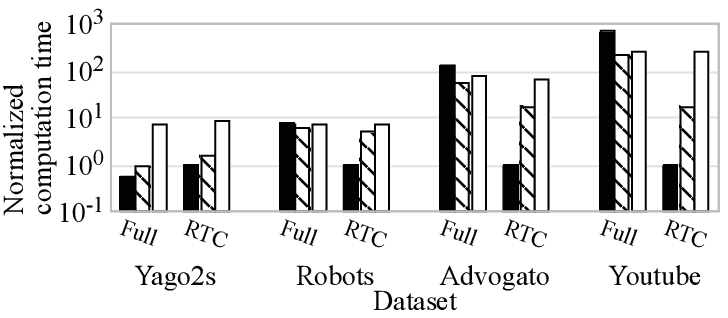}
\caption{Real datasets.}
\end{subfigure}
\vspace{-0.1cm}
\caption{Computation time of three parts of \textit{Full} and \textit{RTC} as the vertex degree is varied ($\sharp$ RPQs = 4).}
\vspace{-0.5cm}
\label{fig: Computation time of three parts on datasets}
\end{figure*} 

\textbf{Query response time} Fig.~\ref{fig: Computational performances varying vertex degrees} shows (a) the query response time on the synthetic datasets and (b) the normalized query response time on the real datasets. Not only the vertex degree but also the size of the dataset (the number of vertices and the number of edges) affect the performance. When experimenting on different real datasets, it is not easy to show the effect of the vertex degree only because the sizes of the datasets are different. To remedy this problem, we compare the relative performance by normalizing the query response time of the comparison methods by that of \textit{RTCSharing} on the real datasets. As shown in Fig.~\ref{fig: Computational performances varying vertex degrees}, \textit{RTCSharing} improves the query response time by 1.63 to 20.20 over \textit{FullSharing} on all graphs except Yago2s. 

Fig.~\ref{fig: Computation time of three parts on datasets} shows the computation time of three detailed parts. The figure confirms that \textit{RTCSharing} beats \textit{FullSharing} through the improvement of Shared\_Data and \textit{Pre\textsubscript{G}} $\bowtie$ \textit{R\textsuperscript{+}\hspace{-0.12cm}\textsubscript{G}}. \textit{RTCSharing} improves Shared\_Data by 7.78 to 9013.64 times and \textit{Pre\textsubscript{G}} $\bowtie$ \textit{R\textsuperscript{+}\hspace{-0.12cm}\textsubscript{G}} by 1.30 to 19.11 times over \textit{FullSharing} on all graphs except Yago2s. Shared\_Data is improved because the computation of \textit{\textoverline{R\textsuperscript{+}\hspace{-0.12cm}\textsubscript{G}}} is simpler than that of \textit{R\textsuperscript{+}\hspace{-0.12cm}\textsubscript{G}}. \textit{Pre\textsubscript{G}} $\bowtie$ \textit{R\textsuperscript{+}\hspace{-0.12cm}\textsubscript{G}} is improved because \textit{RTCSharing} eliminates \textit{redundant-1}, \textit{redundant-2} and \textit{useless-1 operations} by processing and optimizing each RPQ as a relational algebra expression. When we compare \textit{RTCSharing} with \textit{NoSharing}, \textit{RTCSharing} significantly improves the performance by up to 73.86 times over \textit{NoSharing} on all graphs except Yago2s. The reason is that \textit{NoSharing} computes \textit{R\textsuperscript{+}\hspace{-0.12cm}\textsubscript{G}} repeatedly for each RPQ while \textit{RTCSharing} avoids the overhead by sharing \textit{\textoverline{R\textsuperscript{+}\hspace{-0.12cm}\textsubscript{G}}} among RPQs.

For Yago2s, however, Fig.~\ref{fig: Computational performances varying vertex degrees}(b) shows that \textit{RTCSharing} is up to 1.36 times slower than \textit{FullSharing} and 1.07 times slower than \textit{NoSharing} in the query response time. Fig.~\ref{fig: Computation time of three parts on datasets}(b) shows that it is slower 1.82 times in Shared\_Data; 1.88 times in \textit{Pre\textsubscript{G}} $\bowtie$ \textit{R\textsuperscript{+}\hspace{-0.12cm}\textsubscript{G}}. Yago2s is an exceptional case, in which vertex degree is extremely small, 0.02, and the average number of vertices in an SCC of \textit{G\textsubscript{R}} is 1.00, which means that the vertex-level reduction is not very effective. Thus, \textit{G\textsubscript{R}} and \textit{\textoverline{G\textsubscript{R}}} are similar (almost the same) in size, so that the computation times of \textit{R\textsuperscript{+}\hspace{-0.12cm}\textsubscript{G}} and \textit{\textoverline{R\textsuperscript{+}\hspace{-0.12cm}\textsubscript{G}}} also are similar. However, \textit{RTCSharing} has the overhead of reducing the graph, and thus, \textit{RTCSharing} is slower than \textit{FullSharing} in Shared\_Data. \textit{RTCSharing} is slower in \textit{Pre\textsubscript{G}} $\bowtie$ \textit{R\textsuperscript{+}\hspace{-0.12cm}\textsubscript{G}} because the average number of vertices in an SCC of \textit{G\textsubscript{R}} is 1.00, and there are few \textit{redundant-1} and \textit{redundant-2 operations} that \textit{RTCSharing} can eliminate; on the other hand, \textit{RTCSharing} has additional join overheads.

In Figs.~\ref{fig: Computational performances varying vertex degrees} and~\ref{fig: Computation time of three parts on datasets}, we note that, as the vertex degree increases, \textit{RTCSharing} improves the performance gradually. In Figs.~\ref{fig: Computational performances varying vertex degrees}(a) and \ref{fig: Computation time of three parts on datasets}(a) on the synthetic datasets, as the vertex degree increases from 2\textsuperscript{-2} to 2\textsuperscript{4}, the ratio of the query response time of \textit{FullSharing} over \textit{RTCSharing} increases from 1.88 to 20.20; that of Shared\_Data from 10.40 to 9013.64; that of \textit{Pre\textsubscript{G}} $\bowtie$ \textit{R\textsuperscript{+}\hspace{-0.12cm}\textsubscript{G}} from 2.03 to 19.11. Also in Figs.~\ref{fig: Computational performances varying vertex degrees}(b) and \ref{fig: Computation time of three parts on datasets}(b) on the real datasets, as the vertex degree increases from 0.02 to 0.52, 2.61, and 11.42, the ratio of the query response time of \textit{FullSharing} over \textit{RTCSharing} increases from 0.74 to 1.63, 2.92, and 4.20; that of Shared\_Data from 0.55 to 7.78, 129.40, and 671.86; that of \textit{Pre\textsubscript{G}} $\bowtie$ \textit{R\textsuperscript{+}\hspace{-0.12cm}\textsubscript{G}} from 0.53 to 1.30, 2.95, and 11.97. The reasons are as follows. As the vertex degree increases, the average number of vertices of \textit{G\textsubscript{R}} reduced to one vertex of \textit{\textoverline{G\textsubscript{R}}} increases, so that the size of reduced graph \textit{\textoverline{G\textsubscript{R}}} decreases. Thus, the computation time for \textit{\textoverline{R\textsuperscript{+}\hspace{-0.12cm}\textsubscript{G}}} becomes smaller, and the ratio of \textit{FullSharing} over \textit{RTCSharing} also becomes large in Shared\_Data. In addition, since the number of vertices in each SCC of \textit{G\textsubscript{R}} increases, the number of \textit{redundant-1} and \textit{redundant-2 operations} that are eliminated in \textit{RTCSharing} increases, and the ratio becomes larger in \textit{Pre\textsubscript{G}} $\bowtie$ \textit{R\textsuperscript{+}\hspace{-0.12cm}\textsubscript{G}}. As a result, the ratio becomes larger in the query response time as well. We note, however, that the ratio is smaller in the query response time than in Shared\_Data or \textit{Pre\textsubscript{G}} $\bowtie$ \textit{R\textsuperscript{+}\hspace{-0.12cm}\textsubscript{G}}. The reason is as follows. As the vertex degree increases, the length of the path satisfying the RPQ becomes longer so that the number of final results increases, making the time required to find (\textit{Pre}$\cdot$\textit{R\textsuperscript{+}}$\cdot$\textit{Post})\textsubscript{\textit{G}} from (\textit{Pre}$\cdot$\textit{R\textsuperscript{+}})\textsubscript{\textit{G}} increase. Therefore, the portion of Remainder, which is largely the same in both methods, in the query response time increases, lowering the ratio of the query response time. When we compare \textit{RTCSharing} with \textit{NoSharing}, as the vertex degree increases, the ratio of \textit{NoSharing} over \textit{RTCSharing} in the query response time increases from 1.68 to 73.86 on the synthetic datasets and from 0.93 to 2.03, 7.13, and 10.45 on the real datasets. The reason is that the overhead of repeated computation of \textit{R\textsuperscript{+}\hspace{-0.12cm}\textsubscript{G}} increases in \textit{NoSharing} as the vertex degree increases.

\begin{figure}[t]
\centering
\includegraphics[scale=1]{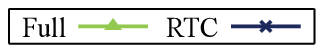}
\begin{subfigure}[b]{.49\columnwidth}
    \includegraphics[scale=1]{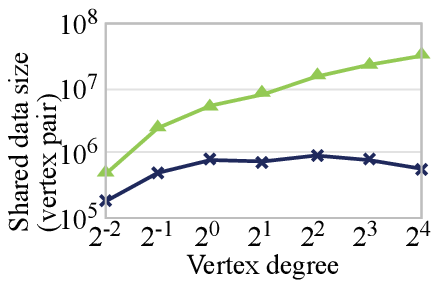}
    \caption{Synthetic datasets.}
\end{subfigure}
\begin{subfigure}[b]{.49\columnwidth}
    \includegraphics[scale=1]{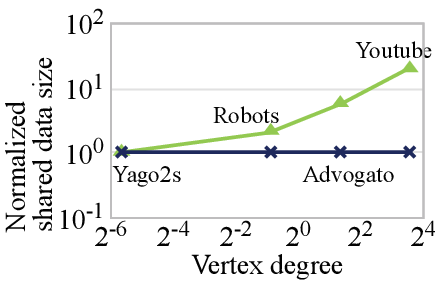}
    \caption{Real datasets.}
\end{subfigure}
\caption{Space performances of \textit{Full} (\textit{R\textsuperscript{+}\hspace{-0.12cm}\textsubscript{G}}) and \textit{RTC} (\textit{\textoverline{R\textsuperscript{+}\hspace{-0.12cm}\textsubscript{G}}}) for datasets as the vertex degree is varied ($\sharp$ RPQs = 4).}
\label{fig: Space performances}
\end{figure}

\textbf{Shared data size} Fig.~\ref{fig: Space performances} shows (a) the shared data size on the synthetic datasets and (b) the normalized shared data size on the real datasets. As the vertex degree increases, the size of \textit{R\textsuperscript{+}\hspace{-0.12cm}\textsubscript{G}} in \textit{FullSharing} increases. On the other hand, even as the vertex degree increases, the size of \textit{\textoverline{R\textsuperscript{+}\hspace{-0.12cm}\textsubscript{G}}} in \textit{RTCSharing} does not increase much. The reason is that as the vertex degree increases, the average number of vertices of \textit{G\textsubscript{R}} reduced to one vertex of \textit{\textoverline{G\textsubscript{R}}} increases, so that the size of reduced graph \textit{\textoverline{G\textsubscript{R}}} decreases (see Fig.~\ref{fig: The number of verticesof comparison methods}). As a result, as the vertex degree increases from 2\textsuperscript{-2} to 2\textsuperscript{4} on the synthetic datasets, the shared data size of \textit{FullSharing} over \textit{RTCSharing} increases from 2.61 to 54.94. On real datasets, as the vertex degree increases from 0.02 to 0.52, 2.61, and 11.42 on the real datasets, that increases from 1.05 to 2.09, 5.87, and 20.23. The reason that the Yago2s' ratio is greater than 1.00 (the number of vertices in an SCC of \textit{G\textsubscript{R}}) is that Yago2s has two exceptional \textit{R}s, which have high vertex degrees.  for them, the size of \textit{R\textsuperscript{+}\hspace{-0.12cm}\textsubscript{G}} over \textit{\textoverline{R\textsuperscript{+}\hspace{-0.12cm}\textsubscript{G}}} is only 1.05.

\begin{figure}[t]
\centering
\includegraphics[scale=1]{Figures/legend_4.eps}
\begin{subfigure}[b]{.49\columnwidth}
    \includegraphics[scale=1]{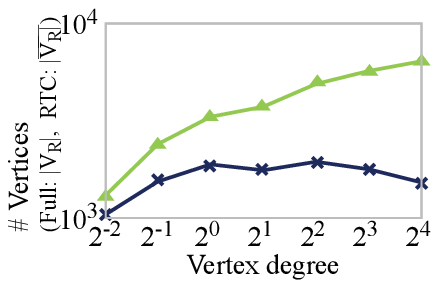}
    \caption{Synthetic datasets.}
\end{subfigure}
\begin{subfigure}[b]{.49\columnwidth}
    \includegraphics[scale=1]{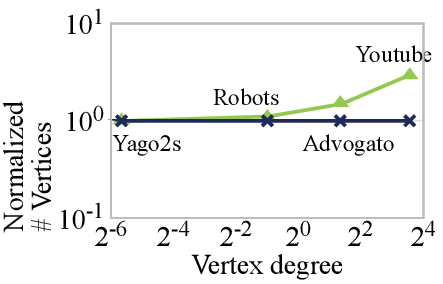}
    \caption{Real datasets.}
\end{subfigure}
\caption{The number of vertices of \textit{Full} ($\mid$\textit{V\textsubscript{R}}$\mid$) and \textit{RTC} ($\mid$\textit{\textoverline{V\textsubscript{R}}}$\mid$) for datasets as the vertex degree is varied ($\sharp$ RPQs = 4).}
\label{fig: The number of verticesof comparison methods}
\end{figure} 

\subsubsection{Experiment 2: The number of RPQs is varied.}
\label{subsubsec:Experiment 2}
In this experiment, we vary the number of RPQs constituting each multiple RPQ set. We use RMAT\_3 and Advogato, which have a median vertex degree among synthetic and real datasets, respectively. We use multiple RPQ sets consisting of 1, 2, 4, 6, 8 and 10 RPQs generated as described in Section~\ref{subsec:Experiment Environment}.

\begin{figure}[b]
\centering
 \includegraphics[scale=1]{Figures/legend_1.eps}
\begin{subfigure}[t]{.49\columnwidth}
    \includegraphics[scale=1]{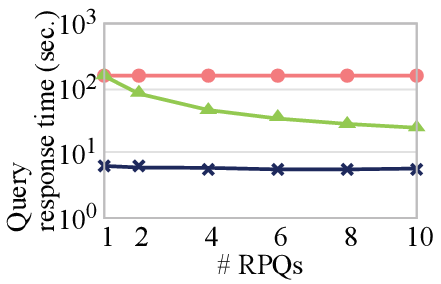}
    \caption{Synthetic dataset (RMAT\_3, vertex degree = 2).}
    \label{fig:Result of Exp 2-1}
\end{subfigure}
\begin{subfigure}[t]{.49\columnwidth}
    \includegraphics[scale=1]{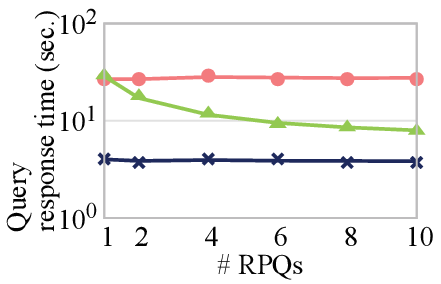}
    \caption{Real dataset (Advogato, vertex degree = 2.61).}
    \label{fig:Result of Exp 2-2}
\end{subfigure}
\caption{Computational performances of \textit{No}, \textit{Full}, and \textit{RTC} as the number of RPQs is varied.}
\label{fig:Computational performances as the number of RPQs is varied}
\end{figure}

Fig.~\ref{fig:Computational performances as the number of RPQs is varied} shows the query response time on (a) a synthetic dataset (RMAT\_3) and (b) a real dataset (Advogato) as the number of RPQs is varied. Fig.~\ref{fig: Computation time of three parts as the number of RPQs is varied} shows the computation time of each part of \textit{RTCSharing} and \textit{FullSharing}. As the number of RPQs increases, Shared\_Data in both \textit{RTCSharing} and \textit{FullSharing} decreases. This is because the computation time of \textit{R\textsuperscript{+}\hspace{-0.12cm}\textsubscript{G}} or \textit{\textoverline{R\textsuperscript{+}\hspace{-0.12cm}\textsubscript{G}}} is amortized by the number of RPQs. In the case of \textit{RTCSharing}, the amortization effect is not large since the portion of Shared\_Data in the total time is much smaller than that of \textit{FullSharing} due to efficiency of \textit{RTCSharing} (see Fig.~\ref{fig: Computation time of three parts as the number of RPQs is varied}). As a result, as the number of RPQs increases from 1 to 10, the ratio of the query response time of \textit{FullSharing} over \textit{RTCSharing} decreases from 24.35 to 4.25 on the synthetic dataset and from 7.17 to 2.08 on the real dataset. The ratio of the query response time of \textit{NoSharing} over \textit{RTCSharing} slightly increases from 23.11 to 25.38 on the synthetic dataset and from 6.76 to 7.17 on the real dataset. This gradual improvement also comes from the amortization effect of Shared\_Data in \textit{RTCSharing} despite that the portion of Shared\_Data in the query response time is small.

\section{Related Work}
\label{sec:Related Work}
\textbf{Transitive closure}
Purdom~\cite{Pur70} and Nuutila~\cite{Nuu94} proposed algorithms for computing transitive closure by using SCCs, which essentially implemented Lemma~\ref{lemma:RTC of GR equals Cartesian product of RTC of GRbar}, but without formalization. Since they are only for an unlabeled graph, they alone cannot be used to evaluate an RPQ. If combined with edge-level reduction, they can be used for evaluating the entire \textit{R\textsuperscript{+}}. However, they cannot efficiently process the RPQ \textit{Pre$\cdot$R\textsuperscript{+}} due to incurring \textit{useless and redundant operations}. Moreover, for multiple RPQs evaluation, they can only facilitate \textit{FullSharing} and cannot directly support \textit{RTCSharing}. 

\begin{figure}[t] 
\centering
\includegraphics[scale=1]{Figures/legend_3.eps}
\begin{subfigure}[t]{\columnwidth}
    \includegraphics[scale=1]{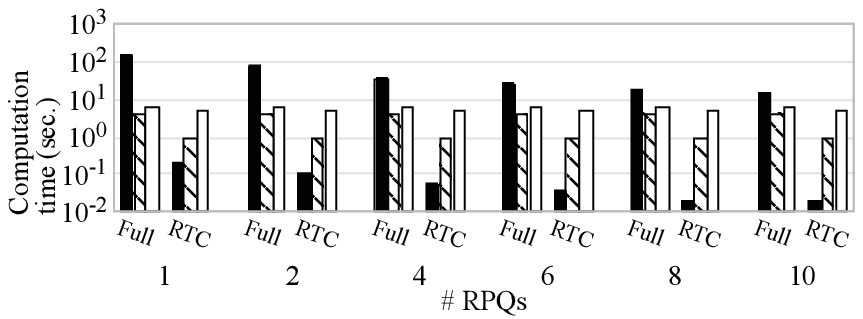}
    \caption{Synthetic dataset (RMAT\_3, vertex degree = 2).}
\end{subfigure}
\begin{subfigure}[t]{\columnwidth}
    \includegraphics[scale=1]{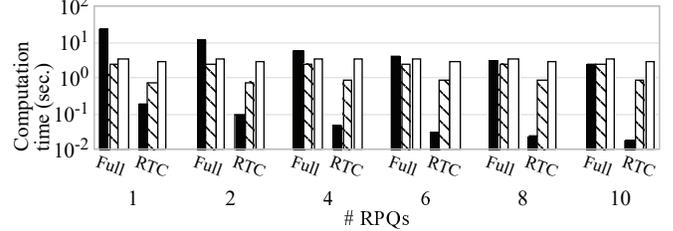}
    \caption{Real dataset (Advogato, vertex degree = 2.61).}
\end{subfigure}
\caption{Computation time of three parts of \textit{Full} and \textit{RTC} as the number of RPQs is varied.}
\label{fig: Computation time of three parts as the number of RPQs is varied}
\end{figure} 

\textbf{Reachability query}
The result of the RPQ \textit{R\textsuperscript{+}} on \textit{G} is the same with the result of the reachability query on \textit{G\textsubscript{R}}. The main costs of this query are the index construction time, the index size, and the query response time. Some of recent approaches~\cite{Jin13,Val17} answer the query using the index only. Some other recent approaches~\cite{Wei18,Su17} traverse a graph at run-time if needed. If combined with edge-level reduction, those methods can also be used for evaluating the entire \textit{R\textsuperscript{+}}. However, since all pairs of vertices O($\mid$\textit{V}$\mid$\textsuperscript{\textit{2}}) should be checked, they cannot be efficient methods. In addition, they cannot avoid \textit{redundant} and \textit{useless operations}.

\textbf{Evaluation methods of single RPQ queries}
A few methods~\cite{Ngu17, Kos12} tried to reduce the number of edges unnecessarily accessed when traversing the graph. The method proposed by Koschmieder and Leser~\cite{Kos12} compares the number of edges of each label that is present in the given RPQ and traverses the edge having the label that has the smallest number of edges first. Then, it continues to traverse other edges only when they are connected to those already traversed. Nguyen and Kim~\cite{Ngu17} considered not only the number of edges of each label but also the number of unique start and end vertices of the edges. 
Yannakakis~\cite{Yan90} proposed a method that reduced the problem of evaluating an RPQ to a problem of finding the paths between given vertices by constructing a graph. This method is entirely different from the edge-level reduced graph used in this paper. Yakovets et al.~\cite{Yak16} proposed a cost model for finding the optimal RPQ evaluation order and a method to efficiently evaluate the Kleene star \textit{R\textsuperscript{*}}. The method is similar to our edge-level graph reduction in that it traverses pre-evaluated paths rather than edges, but does not consider vertex-level graph reduction nor evaluation of multiple RPQs. Fletcher et al.~\cite{Fle16} and Tetzel et al.~\cite{Tet17} proposed methods to evaluate RPQs using indexes. The former finds all the paths in the graph of the lengths less than a certain threshold in advance, creates an index using the path labels of these paths as the keys, and uses the index to evaluate an RPQ. The latter compares the performance of the method using the compressed index and the one using the uncompressed index. Pacaci et al.~\cite{Pac20} focused on the evaluation over streaming graphs. All of these existing methods did not consider \textit{redundant} and \textit{useless operations} and focused mainly on single RPQ evaluation. 

\textbf{Evaluation methods of multiple RPQs} 
Abul-Basher~\cite{Abu17} proposed an optimization technique for evaluating multiple RPQs. This method finds a common sub-query of given multiple RPQs, evaluates it first, and shares the results among the RPQs to avoid repeated computations. This is the \textit{FullSharing} method we used for comparison in Section~\ref{sec:Performance Evaluation}. However, as explained in Section~\ref{sec:Performance Evaluation}, when the common sub-query is \textit{R\textsuperscript{+}}, its evaluation is costly. There are also \textit{redundant} and \textit{useless operations} when evaluating each RPQ using \textit{R\textsuperscript{+}\hspace{-0.12cm}\textsubscript{G}}.

\section{Conclusions}
\label{sec:Conclusion}
In this paper, we propose a notion of \textit{RPQ-based graph reduction} that replaces the evaluation of the Kleene closure on the large original graph \textit{G} to that of the transitive closure to the small graph \textit{\textoverline{G\textsubscript{R}}}. We showed that \textit{R\textsuperscript{+}\hspace{-0.12cm}\textsubscript{G}} can be easily calculated from the transitive closure of \textit{\textoverline{G\textsubscript{R}}} (i.e., RTC), which is computationally simpler and smaller than \textit{R\textsuperscript{+}\hspace{-0.12cm}\textsubscript{G}} and \textit{R\textsuperscript{*}\hspace{-0.12cm}\textsubscript{G}}, in Theorem \ref{theorem:R+G equals Cartesian product of RTC of GRbar}. We also proposed an RPQ evaluation algorithm, \textit{RTCSharing}, that takes advantage of this RTC. \textit{RTCSharing} treats each clause in the DNF of the given RPQ as a batch unit that is in the form of \textit{Prefix$\cdot$R\textsuperscript{+}$\cdot$Postfix} or \textit{Prefix$\cdot$R\textsuperscript{*}$\cdot$Postfix}. In \textit{RTCSharing}, we represent the batch unit as a relational algebra expression (join sequence) including the RTC and efficiently evaluate it sharing the RTC among batch units. We also eliminate \textit{useless-1 operations} by evaluating \textit{R\textsuperscript{+}} only starting from vertex pairs in \textit{Prefix\textsubscript{G}}, \textit{redundant-1} and \textit{redundant-2 operations} by unioning on the intermediate result of each join step, and \textit{useless-2 operations} by using the mutually disjoint property of SCCs. We formally explain that \textit{useless-1}, \textit{redundant-1}, and \textit{redundant-2 operations} are caused by having \textit{Prefix}, and \textit{useless-2 operation} is caused by structural property of \textit{\textoverline{R\textsuperscript{+}\hspace{-0.12cm}\textsubscript{G}}}. Experiments using synthetic and real datasets show that \textit{RTCSharing} significantly improves the performance by up to 73.86 times over \textit{NoSharing} and 20.20 times over \textit{FullSharing} in terms of average query response time.

\section*{Acknowledgment}
This work was partially supported by National Research Foundation (NRF) of Korea grant funded by Korean Government (MSIT) (No. 2016R1A2B4015929) and by Institute of Information \& communications Technology Planning \& evaluation (IITP) grant funded by the Korean Government (MSIT) (No. 2021-0-00859, Development of A Distributed Graph DBMS for Intelligent Processing of Big Graphs). The authors deeply appreciate anonymous reviewers’ incisive and constructive comments that helped make this paper much more readable and complete.


\begin{thebibliography}{00}

\bibitem{Men95}
Mendelzon, O. and Wood, P.,
``Finding regular simple paths in graph databases,''
\textit{SIAM Journal on Computing}, 
Vol. 24, No. 6, pp. 1235-1258, 1995.

\bibitem{Spa13}
Harris, S. and Seaborne, A., SPARQL 1.1 query language, W3C, Mar. 2013. 

\bibitem{Cyp18}
Cypher, http://neo4j.com/ accessed on Dec. 20th, 2018.

\bibitem{Ngu17}
Nguyen, V. and Kim, K.,
``Efficient regular path query evaluation by splitting with unit-subquery cost matrix,''
\textit{IEICE Transactions on Information and Systems}, 
Vol. 100, No. 10, pp. 2648-2652, 2017.

\bibitem{Yak16}
Yakovets, N., Godfrey, P., and Gryz, J.,
``Query planning for evaluating SPARQL property paths,''
In {\em Proc. ACM Int'l Conf. on Management of Data (SIGMOD)},
pp. 1875-1889, San Francisco, California, June-July 2016.

\bibitem{Gra03a}
Grahne, G. and Thomo, A. 
``Query containment and rewriting using views for regular path queries under constraints,''
In {\em Proc. ACM Symposium on Principles of Database Systems (PODS)}, 
pp. 111-122, San Diego, California, June 2003.

\bibitem{Gra03b}
Grahne, G. and Thomo, A. 
``New rewritings and optimizations for regular path queries,''
In {\em Proc. IEEE Int'l Conf. on Database Theory (ICDT)}, 
pp. 242-258, Siena, Italy, Jan. 2003.

\bibitem{Abu17}
Abul-Basher, Z., 
``Multiple-query optimization of regular path queries,''
In {\em Proc. IEEE Int'l Conf. on Data Engineering(ICDE)},
pp. 1426-1430, San Diego, California, Apr. 2017.

\bibitem{Fle16}
Fletcher, G., Peters, J., and Poulovassilis, A., 
``Efficient regular path query evaluation using path indexes,''
In {\em Int'l Conf. on Extending Database Techonology(EDBT)},
pp. 636-639, Bordeaux, France, Mar. 2016.

\bibitem{Kos12}
Koschmieder, A. and Leser, U.,
``Regular path queries on large graphs,''
In \textit{Proc. Int'l Conf. on Scientific and Statistical Database Management(SSDBM)}, 
pp. 177-194, Chania, Greece, June 2012.

\bibitem{Yak15}
Yakovets, N., Godfrey, P., and Gryz, J.,
``Waveguide: Evaluating SPARQL Property Path Queries,''
In {\em Proc. Int'l Conf. on Extending Database Techonology(EDBT)},
pp. 525-528, Brussels, Belgium, Mar. 2015.

\bibitem{Pur70}
Purdom, P.,
``A transitive closure algorithm,''
\textit{BIT Numerical Mathematics}, 
Vol. 10, No. 1, pp. 76-94, 1970.

\bibitem{Nuu94}
Nuutila, E., 
``An efficient transitive closure algorithm for cyclic digraphs,''
\textit{Information Processing Letters}, 
Vol. 52, No. 4, pp. 207-213, Nov. 1994.

\bibitem{Tar72}
Tarjan, R.,
``Depth first search and linear graph algorithms,''
\textit{SIAM Journal on Computing},
Vol. 1, No. 2, pp. 146-160, 1972.

\bibitem{Dav90}
Davey, A. B. and Priestley, A. H., \textit{Introduction to Lattices and Order}, Cambridge University Press, 1990.

\bibitem{Ull88}
Ullman, J., \textit{Principles of database and knowledge-base systems}, Vol. 1, Computer Science Press, 1988.

\bibitem{Cha04}
Chakrabarti, D., Zhan, Y., and Faloutsos, C., ``R-MAT: A recursive model for graph mining,'' 
In \textit{Proc. SIAM SDM},
pp.442-446, 2004.

\bibitem{Par17}
Park, H and Kim, M., "TrillionG: A trillion-scale synthetic graph generator using a recursive vector model," 
In \textit{Proc. ACM SIGMOD}, 
pp. 913-928, 2017.

\bibitem{Yag18}
Yago2s dataset, https://www.mpi-inf.mpg.de/departments/databases-and-information-systems/research/yago-naga/yago/downloads/ accessed on Dec. 2018.

\bibitem{Rob18}
Robots dataset, http://www.trustlet.org/datasets/robots\_net/ accessed on Dec. 2018.

\bibitem{Adv18}
Advogato network dataset, http://konect.uni-koblenz.de/networks/advogato accessed on Dec. 2018.

\bibitem{You18}
Youtube dataset, http://socialcomputing.asu.edu/pages/datasets accessed on Dec. 2018.

\bibitem{Jin13}
Jin, R. and Wang, G.,
``Simple, fast, and scalable reachability oracle,''
\textit{Proceedings of the VLDB Endowment},
Vol. 6, No. 14, pp. 1978-1989, 2013.

\bibitem{Val17}
Valstar, L., Fletcher, G., and Yoshida, Y., 
``Landmark indexing for evaluation of label-constrained reachability queries,''
In {\em Proc. ACM Int'l Conf. on Management of Data (SIGMOD)},
pp. 345-358, Chicago, Illinois, May 2017.

\bibitem{Wei18}
Wei, H., Yu, J. X., Lu, C., and Jin, R., 
``Reachability Querying: An Independent Permutation Labeling Approach,''
\textit{The VLDB Journal}, 
Vol. 27, No. 1, pp. 1-26, 2018.

\bibitem{Su17}
Su, J., Zhu, Q., Wei, H., and Yu, J. X.
``Reachability Querying: Can it be even faster?,''
\textit{IEEE Transactions on Kwowledge and Data Engineering},
Vol. 29, No. 3, pp. 683-697, 2017.

\bibitem{Yan90}
Yannakakis, M.,
``Graph-theoretic methods in database theory,''
In {\em Proc. ACM Symposium on Principles of Database Systems (PODS)},
pp. 230-242, Nashville, Tennessee, Apr. 1990.

\bibitem{Tet17}
Tetzel, F., Hannes, Vo., Paradies, M., and Lehner, W., ``An analysis of the feasibility of graph compression techniques for indexing regular path queries,'' 
In {\em Proc. ACM Int'l Workshop on Graph Data-management Experiences and Systems (GRADES)},
pp.11-16, Chicago, IL., 2017.

\bibitem{Pac20}
Pacaci, A., Bonifati, A., and Ozsu, T. M.
``Regular Path Query Evaluation on Streaming Graphs,''
In {\em Proc. ACM Int'l Conf. on Management of Data},
pp. 1415-1430, Portland, OR, June 2020.


\end{thebibliography}
\end{document}